\documentclass[aps,pra,10pt,superscriptaddress,twocolumn]{revtex4-2}
\usepackage{graphicx}
\usepackage{amsmath,amsthm,amssymb,dsfont,soul,enumitem}
\usepackage{mdframed}
\usepackage{orcidlink}
\usepackage{tikz}
\usetikzlibrary{arrows.meta,positioning,calc,decorations.pathreplacing}
\newcommand{\ket}[1]{\left| #1 \right\rangle}
\newcommand{\bra}[1]{\left\langle #1 \right|}
\newcommand{\proj}[1]{\ket{#1}\hskip-1mm\bra{#1}}
\newcommand{\av}[1]{\langle #1\rangle}

\usepackage{hyperref}
\usepackage[capitalise]{cleveref}
\crefname{equation}{Eq.}{Eqs.}
\crefname{figure}{Fig.}{Figs.}

\newtheorem{theorem}{Theorem}
\newtheorem{lemma}{Lemma}
\newtheorem{proposition}{Proposition}
\newtheorem{corollary}{Corollary}
\theoremstyle{definition}

\crefname{theorem}{Theorem}{Theorems}
\crefname{lemma}{Lemma}{Lemmas}
\crefname{proposition}{Proposition}{Propositions}
\crefname{corollary}{Corollary}{Corollaries}
\crefname{definition}{Definition}{Definitions}
\crefname{table}{Table}{Tables}
\crefname{appendix}{Appendix}{Appendices}
\newcommand{\RoM}{\mathcal{R}}
\newcommand{\stab}{\mathrm{STAB}}
\newcommand{\EC}{\mathcal{E}_{C}}
\newcommand{\Sg}{\mathcal{S}}
\newcommand{\Nor}{N(\mathcal{S})}
\newcommand{\Pa}{\mathcal{P}}
\newcommand{\renc}{\rho_{\mathrm{enc}}}
\newcommand{\Tr}{\operatorname{Tr}}
\newcommand{\Fq}{\mathbb{F}_2}
\newcommand{\capC}{\mathcal{C}}

\begin{document}

\title{Coherent errors, syndrome extraction and the Robustness of Magic in stabiliser codes}

\author{Enrico Bozzetto}
\affiliation{DET, Politecnico di Torino, Corso Duca degli Abruzzi, 24, 10129 Torino, Italy}
\affiliation{Quantum Group, School of Computing, Newcastle University, 1 Science Square, Newcastle upon Tyne, NE4 5TG, UK}
\author{Jonte R. Hance\,\orcidlink{0000-0001-8587-7618}}
\email{jonte.hance@newcastle.ac.uk}
\affiliation{Quantum Group, School of Computing, Newcastle University, 1 Science Square, Newcastle upon Tyne, NE4 5TG, UK}

\begin{abstract}
Coherent errors leave stabiliser-code-encoded magic states in superpositions of syndrome sectors. We ask how much magic these coherent errors add, and whether syndrome extraction depletes logical magic. Using Robustness of Magic, we show that, for a unitary error $U$ on a correctable region which supports no stabiliser, the encoded state has the same robustness as the logical state plus one copy of the Choi state of $U$, in every stabiliser code encoding. Syndrome extraction removes this copy, returning the logical robustness. The removed magic is positive for every non-Clifford error: adding a non-stabiliser state to any state multiplies one plus the robustness by a (tightly-bounded) factor above one. The magic capacity of two-qubit gates Clifford-equivalent to diagonal or single-qubit gates is equal to their Choi-state robustness. We give a linear-algebraic test for when a Clifford error makes syndrome extraction measure a logical operator, and give magic lifetimes of stored single-qubit magic states under Pauli noise. Choosing the plane of the stored state extends these lifetimes about 1.5x on rotated surface codes under low code-capacity depolarising noise. We identify states and channels in odd-prime dimension for which Wigner negativity and robustness show equal lifetimes.
\end{abstract}

\maketitle

\section{Introduction}
\label{sec:intro}

Fault-tolerant quantum computers built on stabiliser codes perform Clifford operations transversally or by lattice surgery~\cite{Horsman2012,Litinski2019}, and obtain universality from magic states consumed by gate teleportation~\cite{GottesmanChuang1999,BravyiKitaev2005,Gidney2024}. Between preparation and use, a magic state sits in a code and is exposed to noise. Under Pauli noise, the encoded state stays a probabilistic mixture of syndrome sectors (the orthogonal subspaces labelled by the syndrome), each holding a mixture of Pauli-rotated copies of the logical state. Coherent errors (e.g., over-rotations and crosstalk~\cite{Kueng2016,Greenbaum2018}) leave the encoded state in a superposition across sectors, meaning syndrome extraction projects the state onto one sector, or dephases the state if information on the corrected error is discarded~\cite{Greenbaum2018,Beale2018,Bravyi2018,Huang2019,Iverson2020,Venn2023}.

For magic states, we ask how much magic the encoded state holds after a coherent error, how much syndrome extraction removes, and whether any of the removed magic is magic of the logical state. Magic in random codes under coherent errors has previously been studied for random Clifford or Haar encoders with a $Z$ rotation on every physical qubit, starting from a logical $\ket{0}$. The stabiliser R\'enyi entropy of the logical state after syndrome measurement then transitions as the error rate or code rate varies~\cite{Niroula2024,SierantTurkeshi2026}; for Haar encoders, the transition is sharp only when we post-select on the trivial syndrome~\cite{SierantTurkeshi2026}. We instead consider a fixed stabiliser code holding a magic logical state, with an error on a correctable region that supports no stabiliser, and compare the magic of the encoded state before and after syndrome extraction. We use Howard and Campbell's Robustness of Magic $\RoM$~\cite{HowardCampbell2017} (\cref{eq:rom}), the smallest $\ell_1$ norm of the coefficients of an expansion of the state into pure stabiliser states, which is 1 for stabiliser mixtures and larger for every other state. It sets the sampling overhead of classical simulation by these quasiprobability decompositions~\cite{HowardCampbell2017,Seddon2021}, and does not increase under the stabiliser operations a decoder uses. We previously showed that, in odd prime dimension, the negative volume of a code-adapted Wigner representation is equal to the syndrome-weighted negativity of the logical states~\cite{BozzettoHance2026}. For qubits, a code-adapted representation still exists, but there is no discrete Hudson theorem (in odd dimension, the only pure states with non-negative Wigner function are the stabiliser states~\cite{Gross2006}), meaning we sadly cannot identify the negativity of the code-adapted representation with logical magic~\cite{BozzettoHance2026}.

Syndrome extraction followed by decoding is a stabiliser operation, and the robustness of the sector-dephased state equals the syndrome-weighted robustness of the decoded logical states (as we show in \cref{lem:sector}, extending the repetition-code result of de Oliveira Junior \emph{et al}~\mbox{\cite{deOliveiraJunior2026}} to every stabiliser code). Our main result concerns a unitary error $U$ on a region $A$ that supports no logical operator and no stabiliser. Such regions are usually called correctable or erasure regions. The encoded state then has the same robustness as $J_U\otimes\rho_L$, the logical state with one copy of the Choi state of $U$ appended, in every such code, while syndrome extraction returns $\RoM(\rho_L)$ (\cref{thm:coherent}). The first statement follows from a code-space version of the canonical form of Fattal \emph{et al}~\cite{Fattal2004} (\cref{lem:bell}), while the second statement follows from the correctability of $A$, combined with \cref{lem:sector}. Both statements hold in every such code, so the robustness that syndrome extraction removes does not depend on the code. (For the stabiliser R\'enyi entropy and pure stabiliser logical states, the analogue of the first statement follows from a result of Hou \emph{et al}~\cite{Hou2026}.) The number of samples needed to simulate a stabiliser circuit by sampling from an optimal decomposition of its input into stabiliser states grows as the square of the robustness of that input~\cite{HowardCampbell2017}, so the first statement also sets this number of samples needed for circuits acting between the error and syndrome extraction. The magic removed, $\RoM(J_U\otimes\rho_L)-\RoM(\rho_L)$, is not magic of the logical state, and we prove that it is positive for every non-Clifford $U$ and every $\rho_L$. This follows because appending any non-stabiliser state $\sigma$ to any state multiplies $\RoM+1$ by at least $\mu(\sigma)>1$, the value of a linear program (\cref{thm:append}). Strict increase already follows where one plus a robustness is multiplicative, as for the generalised robustness of coherence~\cite{Zhu2017} and for some classes of entangled~\cite{VidalTarrach1999,Datta2009,RubboliTomamichel2024} and magic~\cite{Seddon2021,Rubboli2024b,Bravyi2019} states, but this multiplicativity fails for magic in general~\cite{Heimendahl2021,Seddon2021}. The nearest general result is Piani's bound on the relative entropy of entanglement, $E(\rho_{XY})\ge E^{\mathcal M}(\rho_X)+E(\rho_Y)$, where $E^{\mathcal M}$ is the relative entropy measured by a restricted class $\mathcal M$ of measurements on $X$~\cite{Piani2009}. Piani gives this bound for any compact convex set of free states that is closed under partial trace and under conditioning on outcomes of measurements in $\mathcal M$, and $E^{\mathcal M}$ is positive on every non-free state when $\mathcal M$ contains an informationally complete measurement. The stabiliser polytope is such a set when $\mathcal M$ is the set of measurements whose elements are non-negative combinations of pure stabiliser projectors. The set $\mathcal M$ contains informationally complete measurements, so appending a non-stabiliser state strictly increases the relative entropy of magic. Piani's theorem has been applied to stabiliser states before, to show that the regularised relative entropy of magic is faithful~\cite{Veitch2014}, and Brand\~ao \emph{et al}'s generalisation covers stabiliser states with Pauli measurements~\cite{Brandao2020}. The version of \cref{eq:append} for the generalised robustness is the analogue of Piani's bound for the max-relative entropy (\cref{sec:append}). We have not found a bound of this kind on a robustness in the literature.

Errors on regions that can hold a logical operator are not covered by \cref{thm:coherent}, apart from rotations about a Pauli with a non-trivial syndrome (\cref{cor:pauli}), and are bounded instead by the magic capacity~\cite{SeddonCampbell2019}, the largest factor by which a gate can increase robustness. We prove that the magic capacity equals the robustness of the Choi state for every single-qubit gate, and for every two-qubit gate that is Clifford-equivalent to a diagonal or single-qubit gate, that is, equal to such a gate up to Cliffords before and after (\cref{prop:capacity1,thm:capacity2}). For a uniform rotation (the same rotation about a Pauli axis on every qubit), \cref{lem:sector} still gives the robustness after syndrome extraction, and on the bit-flip and $[[4,2,2]]$ codes we find in closed form that the decoded logical state can hold more or less magic than before the error (\cref{sec:clifford}). For Clifford errors on every qubit, syndrome extraction measures the conjugated stabiliser group, so by the rule for code deformation~\cite{Vuillot2019} it measures a logical operator when that group contains a non-trivial logical operator of the code; we write this as an $\Fq$-linear criterion (\cref{prop:measure,prop:F2}). On rotated surface codes of odd distance it shows that layers of $\pi/4$ rotations about $ZZ$, or about $XX$, do not measure a logical operator. Neither do full layers of controlled-$Z$ gates at distances 3, 5 and 7, while every horizontal or vertical CNOT circuit we tested does. For Pauli noise, the sector identity gives explicit magic lifetimes of single-qubit magic states stored in a code, the number of rounds until the state becomes a stabiliser mixture. When logical $Y$ errors are the least likely, as on rotated surface codes under low code-capacity depolarising noise, storing the state in the plane of the logical $X$ and $Z$ axes extends its lifetime by a constant factor (\cref{sec:memories}). Under a uniform rotation repeated in every round, the logical channel of the bit-flip code is a mixture of logical rotations: the robustness decays at the rate of the Pauli-twirled noise to leading order, but the rotation makes it oscillate, so that a stored state can regain magic after entering the stabiliser polytope and a stored stabiliser state gains magic (\cref{sec:memories}). For each odd prime dimension, we identify a state and a class of channels for which lifetimes measured by Wigner negativity and by robustness agree (\cref{sec:odd}).

This paper is laid out as follows. \Cref{sec:setting} sets up the sector identity. \Cref{sec:coherent} shows the effect of unitary errors and channels on correctable regions, with and without stabilisers inside the region (\cref{thm:coherent,prop:channel,cor:stabA}), and rotations about a Pauli with a non-trivial syndrome, and \cref{sec:append} proves the appending bound that makes the removed magic positive. \Cref{sec:capacity} bounds other errors on regions that can hold a logical operator through the magic capacity, and \cref{sec:clifford} treats errors on every qubit: uniform rotations, and Clifford errors after which syndrome extraction measures a logical operator. \Cref{sec:memories} gives magic lifetimes of stored single-qubit magic states under Pauli noise and under repeated uniform rotations, and \cref{sec:odd} extends \cref{thm:coherent} to odd prime dimension, and compares lifetimes defined there by Wigner negativity and by robustness.

\section{Setting and the sector identity}
\label{sec:setting}

We consider an $[[n,k]]$ stabiliser code with stabiliser group $\Sg$ (signs included), normaliser $\Nor$, and a Clifford encoding isometry $V_0=U_{\mathrm{enc}}(\,\cdot\otimes\ket{0^{n-k}})$. For each syndrome $s$ fix a Pauli $E_s$ with that syndrome (the decoder's correction), and let $V_s=E_sV_0$, so that the sector projectors are $\Pi_s=V_sV_s^\dagger$ and $\sum_s\Pi_s=\mathds{1}$. For a physical state $\rho$, the syndrome probabilities are $P(s)=\Tr(\Pi_s\rho)$, the decoded logical states are $\rho_L(s)=V_s^\dagger\rho V_s/P(s)$, and syndrome extraction with the record discarded is the dephasing $\EC(\rho)=\sum_s\Pi_s\rho\Pi_s$.

The Robustness of Magic of an $n$-qubit state is~\cite{HowardCampbell2017}
\begin{equation}
\RoM(\rho)=\min\Big\{\sum_i|x_i|:\ \rho=\sum_ix_i\sigma_i,\ \sigma_i\in\stab_n\Big\},
\label{eq:rom}
\end{equation}
with $\stab_n$ the pure $n$-qubit stabiliser states. This equals $1$ iff $\rho$ lies in the stabiliser polytope, the convex hull of $\stab_n$, whose elements we call stabiliser mixtures. It is submultiplicative, $\RoM(\rho_1\otimes\rho_2)\le\RoM(\rho_1)\RoM(\rho_2)$, convex, non-increasing under stabiliser channels, and non-increasing on average under stabiliser measurements~\cite{HowardCampbell2017,Regula2018}. By linear-programming duality~\cite{HowardCampbell2017},
\begin{equation}
\RoM(\rho)=\max\{\Tr(\Omega\rho):\ |\Tr(\Omega\sigma)|\le1\ \forall\sigma\in\stab_n\},
\label{eq:dual}
\end{equation}
with $\Omega$ ranging over Hermitian operators. For one qubit with Bloch vector $\mathbf{r}$, $\RoM(\rho)=\max(1,|r_x|+|r_y|+|r_z|)$~\cite{Seddon2021}. Following Ref.~\cite{HowardCampbell2017} we write $\ket{H}=(\ket0+e^{i\pi/4}\ket1)/\sqrt2$, with Bloch vector $(1,1,0)/\sqrt2$, which is Clifford-equivalent to Ref.~\cite{BravyiKitaev2005}'s Hadamard eigenstate, and write $\ket{F}$ for the face state with Bloch vector $(1,1,1)/\sqrt3$, Ref.~\cite{BravyiKitaev2005}'s $T$-type magic state. These have $\RoM=\sqrt2$ and $\sqrt3$ respectively.

\begin{lemma}[Sector identity]
\label{lem:sector}
For every physical state $\rho$,
\begin{equation}
\RoM(\EC\rho)=\sum_sP(s)\,\RoM(\rho_L(s))\le\RoM(\rho).
\label{eq:sector}
\end{equation}
In particular $\RoM(\rho)=\sum_sP(s)\RoM(\rho_L(s))$ after any Pauli channel acting on an encoded state.
\end{lemma}

\begin{proof}
Projecting part of a pure stabiliser state onto a computational basis state either gives zero or a multiple of a pure stabiliser state, and Pauli and Clifford conjugation preserve $\stab$. Therefore, $V_s^\dagger\sigma V_s=c\,\tau$ with $c\in[0,1]$ and $\tau$ a pure stabiliser state, for every $\sigma\in\stab_n$, and $V_s\tau V_s^\dagger\in\stab_n$ for $\tau\in\stab_k$. This second fact means $\RoM(V_s\rho_LV_s^\dagger)\le\RoM(\rho_L)$, and pulling back an optimal decomposition with the first gives the reverse inequality, so $V_s$ preserves $\RoM$. Pulling back an optimal decomposition $\rho=\sum_ix_i\sigma_i$ into every sector at once gives $P(s)\rho_L(s)=\sum_ix_ic_{s,i}\tau_{s,i}$ with $\sum_sc_{s,i}=1$, so $\sum_sP(s)\RoM(\rho_L(s))\le\sum_i|x_i|=\RoM(\rho)$. Applied to $\EC\rho$, which has the same $P(s)$ and $\rho_L(s)$, this means $\sum_sP(s)\RoM(\rho_L(s))\le\RoM(\EC\rho)$, and convexity implies the reverse. A Pauli maps the code space into a single sector, so after a Pauli channel $\EC\rho=\rho$.
\end{proof}

\Cref{lem:sector} is the additivity of robustness over a classical flag, here the syndrome sector, which holds in every dimension. For the $N$-qubit repetition code with stabilisers $Z_jZ_{j+1}$ (the bit-flip code for $N=3$), the sector-dephased states are the X-states. For these states de Oliveira Junior \emph{et al} proved the equality in \mbox{\cref{eq:sector}}, which with the one-qubit formula gives $\RoM$ in closed form, and the bound $\RoM\le\sqrt3$ (their Eq.~(3) and Corollary~6)~\mbox{\cite{deOliveiraJunior2026}}. They also used the inequality, with $\EC$ written as a twirl over the stabilisers, to bound the robustness of many-body states from below (their Eq.~(6)), and asked whether the stabilisers of other codes limit magic in the same way. Every state that commutes with the stabilisers of an $[[n,k]]$ code has $\EC\rho=\rho$, so by \mbox{\cref{eq:sector}} its robustness is at most the largest robustness of a $k$-qubit state, which is $\sqrt3$ for $k=1$. The difference $\RoM(\rho)-\RoM(\EC\rho)$ is the robustness held in coherences between sectors. This vanishes under Pauli noise, and \cref{sec:coherent,sec:append,sec:capacity,sec:clifford} study it for coherent errors.

\section{Coherent errors on a correctable region}
\label{sec:coherent}

We first consider a coherent error $U$ on a region $A$ that supports no logical operator and no stabiliser. Such regions, usually called correctable regions, exist in every stabiliser code. They are regions which are smaller than the minimum weight of the code's stabilisers and smaller than the code distance. In such regions neither coherent errors nor logical operators can fit. Write $U=\sum_{P\in\Pa_A}u_PP$ for its expansion over the Paulis supported on $A$, and let $\renc=V_0\rho_LV_0^\dagger$. We write $J_U=(\mathds{1}\otimes U)\proj{\Phi}(\mathds{1}\otimes U^\dagger)$ for the Choi state of $U$, a state of $2|A|$ qubits, with $\ket{\Phi}$ a product of $|A|$ Bell pairs. The condition we need is
\begin{equation*}
\Nor\cap\Pa_A=\{\mathds{1}\}\quad\text{(up to phases)}.
\tag{C}
\label{eq:C}
\end{equation*}
This says that the erasure of $A$ is correctable and that no stabiliser other than the identity is supported inside $A$. Every region with $|A|<d$, where $d$ is the distance of the code, satisfies the first part, and every region of a non-degenerate code with $|A|<d$ satisfies both. Equivalently, every code state is maximally mixed on $A$, since a Pauli has non-zero expectation in some code state if and only if it lies in $\Nor$.

\begin{theorem}[Coherent error on a correctable region]
\label{thm:coherent}
Under \eqref{eq:C}, for every logical state $\rho_L$ and every unitary $U$ on $A$,
\begin{equation}
\begin{aligned}
\RoM(U\renc U^\dagger)&=\RoM(J_U\otimes\rho_L),\\
\RoM(\EC(U\renc U^\dagger))&=\RoM(\rho_L).
\end{aligned}
\label{eq:main}
\end{equation}
The second identity needs only that the erasure of $A$ is correctable. The robustness removed by syndrome extraction, $\Delta(U,\rho_L)=\RoM(J_U\otimes\rho_L)-\RoM(\rho_L)$, satisfies
\begin{multline}
\max\big[\RoM(J_U)-\RoM(\rho_L),\ (\RoM(\rho_L)+1)(\mu(J_U)-1)\big]\\
\le\Delta(U,\rho_L)\le(\RoM(J_U)-1)\,\RoM(\rho_L),
\label{eq:gapbounds}
\end{multline}
with $\mu(J_U)\ge1$ the value of the linear program defined in \cref{thm:append}. The upper bound is attained for stabiliser $\rho_L$, and $\Delta(U,\rho_L)>0$ for every non-Clifford $U$ and every $\rho_L$.
\end{theorem}

A coherent error on a region satisfying \eqref{eq:C} therefore adds one copy of its Choi state beside the logical state, while syndrome extraction removes that copy and nothing else (\cref{fig:schematic}). The magic it removes was never magic of the logical state. It is lost to any procedure that begins by extracting the syndrome, whether or not the record is kept, since every sector decodes to $\rho_L$ up to a logical Pauli. The amount removed is the same for every code in which $A$ satisfies \eqref{eq:C}.

The region $A$ need not be connected. For errors $U_1,\dots,U_\ell$ on disjoint regions whose union satisfies \eqref{eq:C}, with $U=U_1\otimes\cdots\otimes U_\ell$, $J_U=J_{U_1}\otimes\cdots\otimes J_{U_\ell}$ up to the order of the qubits, so the state after the error has the robustness of $\rho_L$ with one Choi state per error appended, and by repeated use of \cref{thm:append} $\RoM(U\renc U^\dagger)+1\ge\prod_i\mu(J_{U_i})\,(\RoM(\rho_L)+1)$, a bound that grows exponentially in $\ell$ for repeated copies of one non-Clifford error. Condition \eqref{eq:C} for each region separately is not enough. On the $[[4,2,2]]$ code, rotations $e^{i\pi X/8}$ of qubits 1 and 2 give $\RoM(U\renc U^\dagger)=1+1/\sqrt2$ for $\rho_L=\ket{H}\otimes\ket{0}$, with $\ket{H}$ on the logical qubit whose $X$ operator is $X_1X_2$, rather than $\RoM(\ket{H}^{\otimes3})=2.219$, because $X_1X_2$ is a logical operator. Syndrome extraction then makes an unsharp measurement of $X_1X_2$, with effects $(3-X_1X_2)/4$ and $(1+X_1X_2)/4$ on the code space, and returns $(1+\sqrt2)/2$ rather than $\RoM(\rho_L)=\sqrt2$.

\begin{figure}[t]
\centering
\begin{tikzpicture}[scale=0.9,every node/.style={font=\small}]
  \draw[rounded corners=3pt,thick] (0,0) rectangle (3.0,1.6);
  \node at (1.5,1.85) {code state $\renc$};
  \fill[orange!25] (0.15,0.15) rectangle (0.95,1.45);
  \draw[orange!70!black,thick] (0.15,0.15) rectangle (0.95,1.45);
  \node at (0.55,0.8) {$A$};
  \node[draw,thick,fill=white,minimum width=0.55cm] (U) at (0.55,-0.75) {$U$};
  \draw[-{Stealth},thick] (U.north) -- (0.55,0.12);
  \node at (2.0,0.8) {$\rho_L$};
  \draw[-{Stealth},thick] (3.25,0.8) -- (4.05,0.8) node[midway,above] {$W$};
  \begin{scope}[xshift=4.3cm]
    \draw[rounded corners=3pt,thick] (0,0) rectangle (3.4,1.6);
    \fill[orange!25] (0.15,0.15) rectangle (1.45,1.45);
    \draw[orange!70!black,thick] (0.15,0.15) rectangle (1.45,1.45);
    \node at (0.8,1.1) {$J_U$};
    \node at (0.8,0.5) {\scriptsize $A\,A'$};
    \node at (2.45,0.8) {$V_0''\rho_LV_0''^\dagger$};
  \end{scope}
  \draw[-{Stealth},thick] (6.0,-0.15) -- (6.0,-0.75);
  \node[align=center] at (6.0,-1.15) {syndrome extraction:\\ $\RoM(J_U\otimes\rho_L)\ \to\ \RoM(\rho_L)$};
\end{tikzpicture}
\caption{Schematic of \cref{thm:coherent}, for a region $A$ satisfying \eqref{eq:C}. A Clifford $W$ acting outside $A$ maps every code state to Bell pairs between $A$ and a set $A'$, next to the logical state encoded in a smaller code with encoder $V_0''$ (\cref{lem:bell}). A unitary $U$ on $A$ turns the Bell pairs into its Choi state $J_U$, so the state after the error has the robustness of $J_U\otimes\rho_L$. Syndrome extraction twirls $U$ into a Pauli mixture and returns $\RoM(\rho_L)$.}
\label{fig:schematic}
\end{figure}
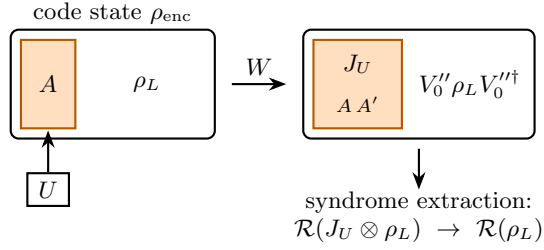

The proof of \cref{thm:coherent} uses two further lemmas.

\begin{lemma}[Syndrome extraction twirls the error]
\label{lem:twirl}
For any operator $O$, $\EC(O)=|\Sg|^{-1}\sum_{g\in\Sg}gOg$. Under \eqref{eq:C}, $\EC(U\renc U^\dagger)=\sum_{P\in\Pa_A}|u_P|^2P\renc P$. If the erasure of $A$ is correctable but $A$ may support stabilisers, $\EC(U\renc U^\dagger)$ is still a mixture $\sum_sw_sP_s\renc P_s$ of Paulis $P_s$ in distinct sectors.
\end{lemma}

\begin{lemma}[Bell-pair normal form]
\label{lem:bell}
Under \eqref{eq:C} there are a set $A'$ of $|A|$ qubits outside $A$, a Clifford $W$ acting outside $A$, and a stabiliser code $\Sg''$ on the remaining $n-2|A|$ qubits with Clifford encoder $V_0''$, such that $(\mathds{1}_A\otimes W)V_0\rho_LV_0^\dagger(\mathds{1}_A\otimes W)^\dagger=\proj{\Phi}_{AA'}\otimes V_0''\rho_LV_0''^\dagger$ for every $\rho_L$.
\end{lemma}

\Cref{lem:twirl} is an exact version, for a region, of the observation that Pauli components of a coherent error with different syndromes are projected onto separate Pauli errors by syndrome extraction~\cite{Greenbaum2018,Iverson2020}. \Cref{lem:bell} is the canonical form of Fattal \emph{et al}~\cite{Fattal2004} (their Theorem~1) for bipartite stabiliser states, extended to code spaces with the Clifford on one side only. For general codes, a unitary on the complement of a correctable region likewise splits every code state into the logical state and a fixed state of the region and part of its complement, not necessarily Bell pairs (condition~(1) of Theorem~3.1 of Ref.~\cite{Harlow2017}). The proofs of \cref{lem:twirl,lem:bell} are in Appendix~\ref{app:bell}.

\begin{proof}[Proof of \cref{thm:coherent}]
By \cref{lem:bell}, $W$ commutes with $U$, so $(\mathds{1}\otimes W)U\renc U^\dagger(\mathds{1}\otimes W)^\dagger=(U\otimes\mathds{1})\proj{\Phi}(U^\dagger\otimes\mathds{1})\otimes V_0''\rho_LV_0''^\dagger$. The first factor is $J_U$ up to the exchange of $A$ and $A'$, a Clifford, and Clifford isometries preserve $\RoM$ (proof of \cref{lem:sector}). The second identity follows from \cref{lem:twirl,lem:sector}: each $P_s\renc P_s$ decodes to $\rho_L$ up to a logical Pauli. The lower bounds in \cref{eq:gapbounds} are monotonicity under partial trace and \cref{thm:append}. The upper bound is submultiplicativity, and for stabiliser $\rho_L$ it is attained because $\RoM(J_U)\le\RoM(J_U\otimes\rho_L)\le\RoM(J_U)$. Positivity follows from \cref{thm:append}: $J_U$ is pure, so it is a stabiliser mixture only if it is a stabiliser state, which happens only when $U$ is Clifford.
\end{proof}

The removed magic can also be isolated before extraction: $W$ depends only on the code and on $A$, so applying it after the error leaves $J_U$ on $A\cup A'$, next to the logical state in the smaller code $\Sg''$. This requires $A$ to be known and reduces the protection of the logical state: for the $[[5,1,3]]$ code and $|A|=2$, the logical qubit is left unencoded. Replacing $J_U$ by fresh Bell pairs and applying $W^\dagger$ returns $\renc$, as in the correction of an erasure, so the protection is lower only between the two applications of $W$. Isolating $J_U$ can be useful when $U$ and $A$ are known, as for a calibrated over-rotation on a known region: before it is replaced, $J_U$ can be used to apply $U$ elsewhere by gate teleportation~\cite{GottesmanChuang1999}, deterministically with Clifford corrections when $U$ is in the third level of the Clifford hierarchy and with heralded success otherwise.

The proof of \cref{thm:coherent} uses only that $W$ acts outside $A$, so it extends to channels.

\begin{proposition}[Channels on the region]
\label{prop:channel}
Under \eqref{eq:C}, for every channel $\mathcal{N}$ on $A$ and every $\rho_L$, \cref{eq:main,eq:gapbounds} hold with $\mathcal{N}(\renc)$ in place of $U\renc U^\dagger$ and $J_{\mathcal{N}}=(\mathrm{id}\otimes\mathcal{N})(\proj{\Phi})$ in place of $J_U$. The excess $\RoM(J_{\mathcal{N}}\otimes\rho_L)-\RoM(\rho_L)$ is positive if and only if $J_{\mathcal{N}}$ is not a stabiliser mixture, that is, if and only if $\mathcal{N}$ is not completely stabiliser-preserving (Theorem~3.1 of Ref.~\cite{SeddonCampbell2019}).
\end{proposition}

\begin{proof}
Write $\mathcal{N}(\cdot)=\sum_iK_i\cdot K_i^\dagger$ with each $K_i$ acting on $A$. Each $K_i$ commutes with $W$, so $(\mathds{1}\otimes W)\mathcal{N}(\renc)(\mathds{1}\otimes W)^\dagger=\sum_i(K_i\otimes\mathds{1})\proj{\Phi}(K_i^\dagger\otimes\mathds{1})\otimes V_0''\rho_LV_0''^\dagger$, and the first factor is $J_{\mathcal{N}}$ up to the exchange of $A$ and $A'$. The proof of \cref{lem:twirl} applies to each $K_i$, so $\EC(\mathcal{N}(\renc))$ is a mixture of Paulis on $A$ applied to $\renc$, and \cref{lem:sector} gives the second identity. The bounds follow as in the proof of \cref{thm:coherent}, and \cref{thm:append} gives a positive excess whenever $J_{\mathcal{N}}$ is not a stabiliser mixture. If it is one, submultiplicativity gives $\RoM(J_{\mathcal{N}}\otimes\rho_L)\le\RoM(\rho_L)$, so the excess is zero.
\end{proof}

A channel that is not completely stabiliser-preserving takes some pure stabiliser state of an ancilla and $A$ outside the stabiliser polytope. Amplitude damping on one qubit of the $[[4,2,2]]$ code, for example, gives a positive excess, while resetting that qubit gives none (accompanying notebook).

The first identity of \cref{eq:main} uses only that $\RoM$ is unchanged by Clifford unitaries and by appending pure stabiliser states, so it holds for every function of the state with these properties, such as the generalised robustness and, for pure states, the stabiliser R\'enyi entropies $M_\alpha$~\cite{Leone2022}. For pure $\rho_L$ every syndrome outcome leaves $\rho_L$ up to a logical Pauli, so the syndrome-averaged $M_\alpha$ of the decoded logical states equals $M_\alpha(\rho_L)$. Refs.~\cite{Niroula2024,SierantTurkeshi2026} study this average for $\alpha=2$ with errors on every qubit, where \eqref{eq:C} fails and the average can exceed $M_2(\rho_L)$. Since $M_\alpha$ is additive, the excess $M_\alpha(J_U\otimes\rho_L)-M_\alpha(\rho_L)=M_\alpha(J_U)$ does not depend on the logical state. For pure stabiliser $\rho_L$, \eqref{eq:C} makes $A$ maximally entangled with the rest of the code block, and the identity for $M_\alpha$ then follows from a result of Hou \emph{et al}: a unitary on a region of a stabiliser state that is maximally entangled with its complement adds $M_\alpha(J_U)$, which equals their unitary stabiliser R\'enyi entropy of $U$~\cite{Hou2026}. The robustness is only submultiplicative, which is why \cref{thm:append} is needed.

\begin{corollary}[Stabilisers in the region]
\label{cor:stabA}
Suppose only that the erasure of $A$ is correctable. Let $\Sg_A=\Sg\cap\Pa_A$ have $a$ independent generators, and let $\Pi_A$ be the projector onto the subspace of $A$ fixed by $\Sg_A$, signs included. Then for every $\rho_L$ and every unitary $U$ on $A$, \cref{eq:main,eq:gapbounds} hold with $J_U^A=2^a(\mathds{1}\otimes U\Pi_A)\proj{\Phi}(\mathds{1}\otimes\Pi_AU^\dagger)$, the Choi state of $U$ restricted to the range of $\Pi_A$, in place of $J_U$, and $\Delta(U,\rho_L)>0$ if and only if $J_U^A$ is not a stabiliser state.
\end{corollary}

\begin{proof}
The image $G_A$ of the restriction map $\Sg\to\Pa_A$ has commutant $\Nor\cap\Pa_A=\Sg_A$ (proof of \cref{lem:twirl}), so it is the commutant of $\Sg_A$ in $\Pa_A$. Choose Paulis $x_j,z_j$ on $A$, $j=1,\dots,|A|-a$, with the commutation relations of $X_j,Z_j$, that commute with $\Sg_A$ and generate $G_A$ together with it. As in the proof of \cref{lem:bell}, the parts outside $A$ of their preimages in $\Sg$ are independent with the same relations, and a Clifford $W$ outside $A$ sends them to $X_{j'},Z_{j'}$ on a set $A'$ of $|A|-a$ qubits. Every stabiliser and every logical operator restricts on $A$ to an element of $G_A$, so it can be multiplied by elements of $\Sg_A$ and by the $x_jX_{j'}$, $z_jZ_{j'}$ to act trivially on $A\cup A'$. Hence $(\mathds{1}\otimes W)\renc(\mathds{1}\otimes W)^\dagger=\proj{\Phi_A}\otimes V_0''\rho_LV_0''^\dagger$, with $\ket{\Phi_A}$ the stabiliser state of $A\cup A'$ with group $\langle\Sg_A,x_jX_{j'},z_jZ_{j'}\rangle$. The stabiliser state $2^{a/2}(\mathds{1}\otimes\Pi_A)\ket{\Phi}$, viewed as a code with no logical qubits, has $\Sg_A$ as its stabilisers supported on $A$, so the same construction, with the same $x_j,z_j$, maps it to $\ket{\Phi_A}$ next to a stabiliser state by a Clifford on the reference. A Clifford isometry on the reference therefore maps $(U\otimes\mathds{1})\ket{\Phi_A}$ to $J_U^A$, and the rest follows as in the proof of \cref{thm:coherent}.
\end{proof}

When $A$ supports no stabiliser, $a=0$ and $J_U^A=J_U$. Otherwise $\RoM(J_U^A)$ can be smaller or larger than $\RoM(J_U)$. On the $[[4,1]]$ code $\Sg=\langle Z_1Z_2,Z_2Z_4,X_1X_2X_3X_4\rangle$, in which the erasure of qubits 1 and 2 is correctable, $T\otimes T$ on these qubits acts on the fixed subspace as a Clifford gate, so $\RoM(U\renc U^\dagger)=\RoM(\rho_L)$, although $\RoM(J_{T\otimes T})=\RoM(\ket{H}^{\otimes2})$. On a bulk $Z$ plaquette of a rotated surface code with $d\ge5$, for the rotation $e^{i\pi Z/8}$ of its four qubits, $J_U^A$ is Clifford-equivalent to $\mathrm{CCZ}\ket{{+}{+}{+}}$, so $\RoM(J_U^A)=23/9=2.556$ (Appendix~\ref{app:numerics}), against $\RoM(J_U)=\RoM(\ket{H}^{\otimes4})=2.863$~\cite{HowardCampbell2017}. For $C^4T$ on a five-qubit region whose erasure is correctable and whose only stabiliser is $+Z_4Z_5$, $J_U^A$ is Clifford-equivalent to the Choi state of $C^3T$, so $\RoM(J_U^A)=2.264>\RoM(J_U)=2.195$ (Appendix~\ref{app:diagonal}).

For rotations about a Pauli outside $\Nor$ the condition on the region can be dropped.

\begin{corollary}[Pauli rotations]
\label{cor:pauli}
Let $U=\exp(i\theta P)$ with $P$ a Pauli outside $\Nor$. Then both identities in \cref{eq:main} hold whatever the support of $P$, and $\RoM(J_U)=|\cos2\theta|+|\sin2\theta|$.
\end{corollary}

\begin{proof}
Choose $g\in\Sg$ anticommuting with $P$ and a Clifford $W'$ with $W'PW'^\dagger=X_1$ and $W'gW'^\dagger=Z_1$. Multiplying the other generators by $g$ where needed, $W'\Sg W'^\dagger=\langle Z_1\rangle\times\Sg'$, the code space becomes $\ket{0}_1\otimes\mathcal{H}'$ with $\mathcal{H}'$ the code space of $\Sg'$, and $W'U\renc U^\dagger W'^\dagger=e^{i\theta X}\proj{0}e^{-i\theta X}\otimes V'\rho_LV'^\dagger$. The state $J_U$ is Clifford-equivalent to $e^{i\theta X}\ket{0}\otimes\ket{0}^{\otimes(2|A|-1)}$, whose Bloch vector has components $\cos2\theta$ and $\sin2\theta$. The cross terms of $U\renc U^\dagger$ are proportional to $P\renc$ and $\renc P$, and $\EC$ removes them because $P\notin\Nor$ maps the code space into another sector. Hence, without \eqref{eq:C}, $\EC(U\renc U^\dagger)=\cos^2\theta\,\renc+\sin^2\theta\,P\renc P$, and $P\renc P$ decodes to $\rho_L$ up to a logical Pauli.
\end{proof}

The bounds of \cref{eq:gapbounds} involve only $J_U$ and $\rho_L$, so under \cref{cor:pauli} they also bound the robustness removed by syndrome extraction. On the three-qubit bit-flip code, for example, $Z_1$ is a logical operator, so \eqref{eq:C} fails for $A=\{1\}$, but \cref{cor:pauli} applies to $\exp(i\pi X_1/8)$. On an encoded $\ket{H}$ this error gives $\RoM=(1+3\sqrt2)/3=1.747547$, the robustness of two copies of $\ket{H}$~\cite{HowardCampbell2017}: the state after the error is Clifford-equivalent to $\ket{H}^{\otimes2}$, and syndrome extraction returns $\RoM(H)=\sqrt2$. \Cref{tab:gap} collects further cases. In every row the accompanying notebook also computes $\RoM(J_U\otimes\rho_L)$ and $\RoM(\EC(U\renc U^\dagger))$ separately, and they agree with $\RoM(U\renc U^\dagger)$ and $\RoM(\rho_L)$, respectively, to solver precision.

\begin{table*}[t]
\caption{Robustness before ($\RoM(U\renc U^\dagger)$) and after ($\RoM(\rho_L)$) syndrome extraction for single-qubit coherent errors on the three-qubit bit-flip code and the $[[4,2,2]]$ code. The Hadamard-axis rotation is $\exp[i\pi(X+Z)/(8\sqrt2)]$, about the axis $(1,0,1)/\sqrt2$. $\Delta$ is the robustness removed, and the bounds on $\Delta$ are those of \cref{eq:gapbounds}, the lower bound being the larger of its two terms, with $\mu(J_U)=1.138071$ for $\exp(i\pi X_1/8)$ and $1.284518$ for the Hadamard-axis rotation. On the bit-flip code \eqref{eq:C} fails for $A=\{1\}$, and these rows follow from \cref{cor:pauli}. Values computed by linear programming over all stabiliser states (accompanying notebook).}
\label{tab:gap}
\begin{ruledtabular}
\begin{tabular}{lllcccccc}
code & error & $\rho_L$ & $\RoM(\rho_L)$ & $\RoM(J_U)$ & $\RoM(U\renc U^\dagger)$ & $\Delta$ & lower bound & upper bound\\
\hline
bit-flip & $\exp(i\pi X_1/8)$ & $\ket{0}$ & 1 & 1.414214 & 1.414214 & 0.414214 & 0.414214 & 0.414214\\
bit-flip & $\exp(i\pi X_1/8)$ & $\ket{H}$ & 1.414214 & 1.414214 & 1.747547 & 0.333333 & 0.333333 & 0.585786\\
bit-flip & $\exp(i\pi X_1/8)$ & $\ket{F}$ & 1.732051 & 1.414214 & 2.109268 & 0.377217 & 0.377217 & 0.717439\\
$[[4,2,2]]$ & $\exp(i\pi X_1/8)$ & $\ket{H}\otimes\ket{H}$ & 1.747547 & 1.414214 & 2.218951 & 0.471405 & 0.379357 & 0.723858\\
$[[4,2,2]]$ & Hadamard-axis & $\ket{H}\otimes\ket{0}$ & 1.414214 & 1.707107 & 2.121320 & 0.707107 & 0.686887 & 1.000000\\
$[[4,2,2]]$ & Hadamard-axis & $\ket{F}\otimes\ket{0}$ & 1.732051 & 1.707107 & 2.509368 & 0.777317 & 0.777317 & 1.224745\\
\end{tabular}
\end{ruledtabular}
\end{table*}

\section{Adding a non-stabiliser state}
\label{sec:append}

\cref{thm:coherent}'s positivity claim rests on the following statement, which holds for any states and any numbers of qubits.

\begin{theorem}
\label{thm:append}
For every $m$-qubit state $\sigma$, and every $n$-qubit state $\rho$,
\begin{equation}
\RoM(\sigma\otimes\rho)+1\ge\mu(\sigma)\,\big(\RoM(\rho)+1\big),
\label{eq:append}
\end{equation}
where
\begin{multline}
\mu(\sigma)=\min\big\{\Tr Y:\ Y\in\mathrm{cone}(\stab_m),\\
\bra{t}Y\ket{t}\ge\bra{t}\sigma\ket{t}\ \forall t\in\stab_m\big\},
\label{eq:mu}
\end{multline}
and $\mathrm{cone}(\stab_m)$ is the set of non-negative combinations of pure stabiliser states. Therefore $\mu(\sigma)\ge1$, with equality iff $\sigma$ is a stabiliser mixture. Hence $\RoM(\sigma\otimes\rho)>\RoM(\rho)$ for every $\sigma$ that is not a stabiliser mixture and every $\rho$.
\end{theorem}

\begin{proof}
Write $r=\RoM(\rho)$ and $R'=\RoM(\sigma\otimes\rho)$, and take an optimal decomposition $\sigma\otimes\rho=\Gamma_+-\Gamma_-$ with $\Gamma_\pm\in\mathrm{cone}(\stab_{m+n})$, $\Tr\Gamma_+=(R'+1)/2$ and $\Tr\Gamma_-=(R'-1)/2$. For $t\in\stab_m$ let $\Gamma_{\pm,t}=(\bra{t}\otimes\mathds{1})\Gamma_\pm(\ket{t}\otimes\mathds{1})$ and $p_t=\bra{t}\sigma\ket{t}$. As in the proof of \cref{lem:sector}, $\Gamma_{\pm,t}\in\mathrm{cone}(\stab_n)$, and $\Gamma_{+,t}-\Gamma_{-,t}=p_t\rho$. If $p_t>0$ this is a decomposition of $\rho$, so $a_t=\Tr(\Gamma_{+,t}+\Gamma_{-,t})\ge rp_t$, which holds trivially if $p_t=0$. Let $Y=2\Tr_n\Gamma_+/(r+1)$. Marginals of stabiliser states are stabiliser mixtures, so $Y\in\mathrm{cone}(\stab_m)$, and $\bra{t}Y\ket{t}=(a_t+p_t)/(r+1)\ge p_t$. Hence $(R'+1)/(r+1)=\Tr Y\ge\mu(\sigma)$. Any feasible $Y'$ in \cref{eq:mu} has $\Tr Y'\ge1$, by summing $\bra{t}Y'\ket{t}\ge p_t$ over a stabiliser basis, so $\mu(\sigma)\ge1$. If $\mu(\sigma)=1$, an optimal $Y'$ has $\Tr Y'=1$, which forces $\bra{t}(Y'-\sigma)\ket{t}=0$ for every $t\in\stab_m$, since every pure stabiliser state belongs to a stabiliser basis. These projectors span the operators, so $\sigma=Y'$, a stabiliser mixture. Conversely $Y'=\sigma$ is feasible when $\sigma$ is one.
\end{proof}

By linear-programming duality, $\mu(\sigma)=\max\{\Tr(\Theta\sigma):\ \Theta\in\mathrm{cone}(\stab_m),\ \bra{\phi}\Theta\ket{\phi}\le1\ \forall\phi\in\stab_m\}$. This gives a second proof of \cref{eq:append}: take $\Theta=\sum_tc_t\proj{t}$ optimal in this maximisation, with $c_t\ge0$, and $\Omega_\rho$ optimal in \cref{eq:dual}. For every $\tau\in\stab_{m+n}$ let $\tau_t=(\bra{t}\otimes\mathds{1})\tau(\ket{t}\otimes\mathds{1})$. Each $\Tr[(\mathds{1}+\Omega_\rho)\tau_t]$ lies in $[0,2\Tr\tau_t]$, and $\sum_tc_t\Tr\tau_t=\Tr(\Theta\Tr_n\tau)\le1$, so $\Omega=\Theta\otimes(\mathds{1}+\Omega_\rho)-\mathds{1}$ is feasible in \cref{eq:dual}. Its value on $\sigma\otimes\rho$ is $\mu(\sigma)(\RoM(\rho)+1)-1$.

\begin{proposition}[One qubit]
\label{prop:muone}
For a single-qubit state with Bloch vector $\mathbf{s}$, $\mu=1+\delta$ with $\delta$ the least non-negative number satisfying $\sum_k\max(0,|s_k|-\delta)\le1+\delta$. In particular $\mu\ge(\RoM+3)/4$, $\mu(F)=(3+\sqrt3)/4$ and $\mu(H)=1+(\sqrt2-1)/3$.
\end{proposition}

\begin{proof}
The cone of the octahedron is $\{(y_0\mathds{1}+\mathbf{y}\cdot\boldsymbol{\sigma})/2:\|\mathbf{y}\|_1\le y_0\}$, and the constraints on the six stabiliser states read $|y_k-s_k|\le y_0-1$. Minimising $y_0=1+\delta$ is solved by $y_k=\mathrm{sgn}(s_k)\max(0,|s_k|-\delta)$.
\end{proof}

The quantity $\mu$ is invariant under Cliffords and unchanged by appending pure stabiliser states. For a Pauli rotation, $J_U$ is Clifford-equivalent to a single-qubit state with two non-zero Bloch components next to a stabiliser state, and \cref{prop:muone} gives $\mu(J_U)-1=(\RoM(J_U)-1)/3$, so that
\begin{equation}
\frac{\RoM(\rho_L)+1}{3}\le\frac{\Delta(U,\rho_L)}{\RoM(J_U)-1}\le\RoM(\rho_L)
\label{eq:pauliratio}
\end{equation}
for every non-Clifford Pauli rotation $U$.

\emph{Tightness.} The bound \eqref{eq:append} is attained. For example $\RoM(F\otimes F)=\RoM(F\otimes F_\perp)=\tfrac12+\sqrt3=(1+\sqrt3)\mu(F)-1$, with $F_\perp$ the state orthogonal to $F$. The value of $\RoM(F\otimes F)$ is given in closed form in Ref.~\cite{HowardCampbell2017} and numerically in Ref.~\cite{HeinrichGross2019}, where $F$ is called $T$. In both cases \cref{thm:append} gives the lower bound, and exact decompositions into stabiliser states, with coefficients in $\mathbb{Q}(\sqrt3)$, give the upper bound. At $\theta=\pi/8$ the lower end of \cref{eq:pauliratio} is attained for $\rho_L=H$ and $F$ (\cref{tab:gap}), with $\Delta=1/3$ and $(1+\sqrt3)(\sqrt2-1)/3$. For $\sigma=H$ and for the Choi states of $\exp(i\pi X/16)$ and of the Hadamard-axis rotation of \cref{tab:gap}, the bound is attained for every $\RoM(\rho)$ from $1.6$ to $\sqrt3$ (Appendix~\ref{app:tightness}).

\Cref{fig:gap} shows $\Delta$ as a function of the rotation angle, with the bounds of \cref{eq:gapbounds}. The excess vanishes at Clifford angles and is largest at $\theta=\pi/8$. It does not follow the size of the coherence between sectors: for the bit-flip code the trace norm of the inter-sector part of $U\renc U^\dagger$ is $\sin2\theta$, which grows up to $\theta=\pi/4$, where $\Delta$ is zero again.

\begin{figure}[t]
\centering
\IfFileExists{fig_gap.pdf}{\includegraphics[width=\columnwidth]{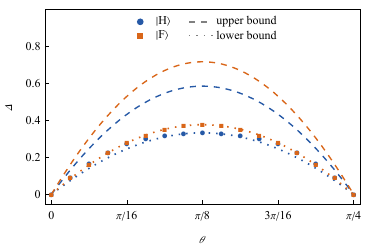}}{\fbox{\parbox{0.9\columnwidth}{\centering\vspace{1.5cm}fig\_gap.pdf (produced by the accompanying notebook)\vspace{1.5cm}}}}
\caption{Robustness $\Delta$ removed by syndrome extraction after a rotation $\exp(i\theta X_1)$, for logical states $\ket{H}$ and $\ket{F}$ (points), with the upper bound $(\RoM(J_U)-1)\RoM(\rho_L)$ and the lower bound $(\RoM(\rho_L)+1)(\RoM(J_U)-1)/3$ (lines) of \cref{eq:gapbounds,eq:pauliratio}. By \cref{cor:pauli} the values are the same for every code with $X_1\notin\Nor$, and they are symmetric about $\theta=\pi/8$, since a Clifford maps $J_U$ at $\theta$ to $J_U$ at $\pi/4-\theta$.}
\label{fig:gap}
\end{figure}

The proof of \cref{thm:append} uses four properties of the stabiliser set. It is closed under marginals and under projecting a subsystem onto a pure stabiliser state, every pure stabiliser state belongs to an orthonormal basis of pure stabiliser states, and these states span the operator space. The same argument therefore applies to the generalised robustness $\RoM_g$~\cite{Steiner2003,Regula2018,Seddon2021}, normalised to vanish on stabiliser mixtures. Writing $\sigma\otimes\rho\le(1+\RoM_g(\sigma\otimes\rho))\tau$ with $\tau$ a stabiliser mixture and projecting the first register onto $t$ gives $p_t(1+\RoM_g(\rho))\le(1+\RoM_g(\sigma\otimes\rho))\Tr\tau_t$, and hence $1+\RoM_g(\sigma\otimes\rho)\ge\mu(\sigma)(1+\RoM_g(\rho))$. The robustness $\RoM_s(\rho)$ in the normalisation of Vidal and Tarrach~\cite{VidalTarrach1999} is the least $s$ for which $(\rho+s\tau)/(1+s)$ is a stabiliser mixture for some stabiliser mixture $\tau$. It satisfies $\RoM=1+2\RoM_s$~\cite{Regula2018}, and \cref{eq:append} becomes $1+\RoM_s(\sigma\otimes\rho)\ge\mu(\sigma)(1+\RoM_s(\rho))$.

For the robustness of entanglement $\RoM_s$, Vidal and Tarrach showed $\RoM_s(\sigma\otimes\rho)\ge\max[\RoM_s(\sigma),\RoM_s(\rho)]$~\cite{VidalTarrach1999}, which allows equality. With pure product states in place of pure stabiliser states and the separable cone in place of $\mathrm{cone}(\stab_m)$ in \cref{eq:mu}, the proof of \cref{thm:append} applies unchanged to the robustness and the generalised robustness of entanglement across the cut $A_1A_2|B_1B_2$, for $\sigma$ on $A_1B_1$ and $\rho$ on $A_2B_2$. Appending any entangled state therefore strictly increases both. Strict increase already follows for pure states, where $1+\RoM_s$ is multiplicative by the pure-state formula of Ref.~\cite{VidalTarrach1999}, and for the generalised robustness when one of the two states is maximally entangled~\cite{Datta2009} or, more generally, maximally correlated, a class that includes pure states (Theorem~9 of Ref.~\cite{RubboliTomamichel2024}, which extends to $\log(1+\RoM_g)$ as the limit $\alpha=z\to\infty$ of its family of measures). For the relative entropy of entanglement it follows from Piani's bound (\cref{sec:intro}) applied to $\sigma\otimes\rho$, since the local measurements include informationally complete ones~\cite{Piani2009}. The $\RoM_g$ case of \cref{eq:append} has the form of Piani's bound with the max-relative entropy in place of the relative entropy. Here $\log(1+\RoM_g)=D_{\max}(\cdot\|\stab)$~\cite{Rubboli2024b}, as for entanglement~\cite{Datta2009}. Writing $Y=c\tau$ in \cref{eq:mu}, with $\tau$ a stabiliser mixture, gives $\log\mu(\sigma)=\min_\tau\max_t\log(\bra t\sigma\ket t/\bra t\tau\ket t)$. This is the cone-restricted max-relative entropy of Ref.~\cite{GeorgeChitambar2024}, with the cone of pure stabiliser states, minimised over stabiliser mixtures $\tau$. Since every pure stabiliser state belongs to a stabiliser basis, it also equals the max-relative entropy measured by measurements whose elements lie in that cone. For entanglement the same holds with the separable cone, for which the cone-restricted and measured max-relative entropies coincide~\cite{Rippchen2025}.

For the Robustness of Magic, $\log[(\RoM+1)/2]$ is not the max-relative entropy to the stabiliser polytope, so this reading does not cover its case. We have not found a version of Piani's bound for a robustness in the literature. Strict increase also follows where $1+\RoM_g$ is multiplicative, as for coherence (Theorem~3 of Ref.~\cite{Zhu2017}). For magic, $1+\RoM_g$ is multiplicative for products of single-qubit states~\cite{Seddon2021,Rubboli2024b} and of a class of depolarised two- and three-qubit states (Theorem~4 of Ref.~\cite{Rubboli2024b}). It is also multiplicative for products of pure states that each have at most three qubits, where it equals the stabiliser extent~\cite{Regula2018,Seddon2021}, which is multiplicative there (Proposition~1 of Ref.~\cite{Bravyi2019}). It is not multiplicative in general~\cite{Heimendahl2021,Seddon2021}. The Robustness of Magic is only submultiplicative. The multiplicative stabiliser norm $\mathcal{D}(\rho)=2^{-n}\sum_P|\Tr(P\rho)|$~\cite{Campbell2011,HowardCampbell2017} gives the lower bound $\RoM(\sigma\otimes\rho)\ge\mathcal{D}(\sigma)\mathcal{D}(\rho)$, which shows strict increase only when $\mathcal{D}(\sigma)\mathcal{D}(\rho)>\RoM(\rho)$. This fails for $\sigma$ close to the stabiliser polytope whenever $\mathcal{D}(\rho)<\RoM(\rho)$.

\section{Magic capacity}
\label{sec:capacity}

A coherent error $U$ on an encoded state can increase its robustness by at most the factor $\capC(U)$, the magic capacity of Seddon and Campbell~\cite{SeddonCampbell2019},
\begin{equation}
\capC(U)=\max_{\sigma\in\stab_{n+m}}\RoM\big((U\otimes\mathds{1}_m)\sigma(U\otimes\mathds{1}_m)^\dagger\big),
\label{eq:capacity}
\end{equation}
with $U$ acting on $n$ qubits and $m\ge n$ ancilla qubits, since $\RoM(U\rho U^\dagger)\le\capC(U)\RoM(\rho)$. Seddon and Campbell showed that an ancilla register of $n$ qubits suffices, that $\RoM(J_U)\le\capC(U)$ because the Choi state is one of the outputs, and that equality holds for gates in the third level of the Clifford hierarchy, which map Paulis to Cliffords under conjugation, such as $T$, CS and CCZ (Theorem~3.2 of Ref.~\cite{SeddonCampbell2019}). Equality does not hold in general: their five-qubit multicontrol-T gate has $\capC(U)=2.264$ against $\RoM(J_U)=2.195$ (Table~III of their supplementary material). For errors on regions satisfying \eqref{eq:C}, \cref{thm:coherent} replaces the bound by the value $\RoM(J_U\otimes\rho_L)$, and \cref{cor:stabA} replaces it by $\RoM(J_U^A\otimes\rho_L)$ when only the erasure of $A$ is correctable. For regions whose erasure is not correctable, apart from rotations about Paulis outside $\Nor$ (\cref{cor:pauli}), the capacity remains the relevant quantity. For gates with $\capC(U)=\RoM(J_U)$, which include those of \cref{prop:capacity1,thm:capacity2} below, the capacity bound gives $\RoM(U\renc U^\dagger)\le\RoM(J_U)\RoM(\rho_L)$ on every region on which $U$ acts. This is the upper bound of \cref{eq:gapbounds}, which a region satisfying \eqref{eq:C} attains for stabiliser $\rho_L$.

\begin{proposition}[One qubit]
\label{prop:capacity1}
For every single-qubit unitary $U$, $\capC(U)=\RoM(J_U)$.
\end{proposition}

\begin{proof}
Let $\sigma$ be a pure stabiliser state on $1+m$ qubits. Restricting its stabiliser group to the first qubit gives a subgroup of $\{\mathds{1},X,Y,Z\}$ (modulo phases) of order 2 or 4, since order 1 would make all $2^{1+m}$ elements of the group act trivially on the first qubit, more than a stabiliser group on $m$ qubits can hold. If the order is 4, no Pauli on the first qubit is a stabiliser, so \eqref{eq:C} holds with $A$ the first qubit for $\sigma$ viewed as a code with no logical qubits, and \cref{lem:bell} gives $(U\otimes\mathds{1})\sigma(U\otimes\mathds{1})^\dagger=J_U\otimes\tau$ up to Cliffords, with $\tau$ a stabiliser state. If the order is 2, $\sigma=\phi\otimes\tau$ with $\phi$ a Pauli eigenstate, and $\RoM(U\phi U^\dagger)\le\RoM(J_U)$ by averaged monotonicity: measuring the reference qubit of $J_U$ in the eigenbasis of the same Pauli leaves $U\phi U^\dagger$ or $U\phi_\perp U^\dagger$ with probability $1/2$ each, and these have opposite Bloch vectors and so equal robustness.
\end{proof}

Numerical results of Seddon and Campbell suggested this for rotations about $X$~\cite{SeddonCampbell2019}. For rotations about a Pauli axis, which are Clifford-equivalent to diagonal gates, it also follows from Theorem~S1 of their supplementary material (see the proof of \cref{thm:capacity2} in Appendix~\ref{app:capacity2}), since the affine subspaces of $\Fq$ are its two points and $\Fq$ itself. For two qubits, the canonical form of Fattal \emph{et al}~\cite{Fattal2004}, with the Clifford on the ancilla only as in \cref{lem:bell}, shows that a pure stabiliser state of the system and an ancilla is, up to Cliffords on the ancilla, of one of two kinds. Either it has two Bell pairs between system and ancilla, which $U$ turns into $J_U$, or it has at most one and is a stabiliser state of the system and one ancilla qubit next to a stabiliser state of the other ancilla qubits. By enumeration, the three-qubit stabiliser states fall into 90 classes under Cliffords on the third qubit, so $\capC(U)$ is the largest of $\RoM(J_U)$ and 90 three-qubit robustness values.

\begin{theorem}[Two qubits]
\label{thm:capacity2}
Let $U$ be a two-qubit unitary. If $U=KDK'$ with $K,K'$ Clifford and $D$ diagonal, then $\capC(U)=\RoM(J_U)=\RoM(D\ket{{+}{+}})$. If $U=K(V\otimes\mathds{1})K'$ with $V$ a single-qubit unitary, then $\capC(U)=\RoM(J_U)=\RoM(J_V)$.
\end{theorem}

The proof, in Appendix~\ref{app:capacity2}, uses Theorem~S1 in the supplementary material of Ref.~\cite{SeddonCampbell2019} and witnesses built from pairwise anticommuting Paulis.

The first case contains every gate $(K_1\otimes K_2)\exp[i(aXX+bYY+cZZ)](K_3\otimes K_4)$ with Clifford $K_j$, since a Clifford maps the commuting pair $XX,ZZ$ to $ZI,IZ$ and so $YY=-XX\cdot ZZ$ to $-ZZ$. The exponential is the non-local part of every two-qubit gate in its canonical decomposition~\cite{KrausCirac2001,Khaneja2001}. The first case also contains the controlled-Hadamard gate, $\mathrm{CH}=e^{-i\pi/4}\,\mathrm{CNOT}\,(\mathds{1}\otimes HSH)\,\mathrm{diag}(1,i,e^{i\pi/4},e^{i\pi/4})\,(\mathds{1}\otimes HS)$. The second case contains, for example, $\exp(i\pi XX/8)\exp(i\pi ZI/8)$, which commutes with the anticommuting pair $IX$, $ZY$ and equals $K(V\otimes\mathds{1})K^\dagger$ with $K=\exp(i\pi ZX/4)$ and $V=\exp(i\pi Y/8)\exp(i\pi Z/8)$. For two-qubit diagonal gates the capacity is reached by the product input $\ket{{+}{+}}$, without an ancilla. Seddon and Campbell observed $\capC(U)=\RoM(J_U)$ numerically for random diagonal gates on two to four qubits~\cite{SeddonCampbell2019} (their Table~1). For the stabiliser nullity, $n-\log_2$ of the number of signed Paulis that stabilise the state, the analogous statement holds for every gate: the largest nullity reachable from stabiliser inputs with ancillas is that of the Choi state (Theorems~22 and 24 of Ref.~\cite{JiangWang2023}). Hou \emph{et al} define a unitary stabiliser R\'enyi entropy and show that it equals $M_\alpha(J_U)$~\cite{Hou2026}; for $M_2$, other stabiliser inputs can give more (Appendix~\ref{app:diagonal}).

The question is open for the remaining two-qubit gates. Of the 17 structured gates of Appendix~\ref{app:numerics}, 12 are covered by \cref{prop:capacity1,thm:capacity2}. For the other five, and for 160 random controlled or multiplexed gates, none of the 90 input classes gives more than $\RoM(J_U)$ to solver precision.

For diagonal gates on $n$ qubits, Theorem~S1 in the supplementary material of Ref.~\cite{SeddonCampbell2019} reduces the capacity to $\capC(U)=\max_{\mathcal{A}}\RoM(U\ket{\mathcal{A}})$ over affine subspaces $\mathcal{A}\subseteq\Fq^n$, with $\ket{\mathcal{A}}$ the uniform superposition over $\mathcal{A}$, and $\mathcal{A}=\Fq^n$ gives $\RoM(J_U)$. Equality with $\RoM(J_U)$ holds for every gate we checked at $n=3$ and $4$. It fails at $n=5$ for the $C^4T$ gate mentioned above, whose capacity is reached by the product input $\ket{+}^{\otimes4}\ket{1}$~\cite{SeddonCampbell2019} (Seddon and Campbell put the phase on $\ket{0^5}$, which moves this input to $\ket{+}^{\otimes4}\ket{0}$; closed forms in Appendix~\ref{app:diagonal}).

\section{Errors on every qubit}
\label{sec:clifford}

When the error acts on every qubit, \eqref{eq:C} fails and syndrome extraction can change the magic of the logical state.

\emph{Uniform rotations.} For a uniform rotation $U=\bigotimes_je^{i\theta P_j}$, with the same angle and Pauli axis on every qubit $j$~\cite{Huang2019,Iverson2020}, $A$ is the whole code block, which supports logical operators, and $U$ is not a single Pauli rotation, so neither \cref{thm:coherent} nor \cref{cor:pauli} applies. \Cref{lem:sector} still gives the robustness after syndrome extraction, $\sum_sP(s)\RoM(\rho_L(s))$ with $\rho_L(s)=K_s\rho_LK_s^\dagger/P(s)$ and logical Kraus operators $K_s=V_s^\dagger UV_0$. Unlike under \eqref{eq:C}, where each $K_s$ is a multiple of a logical Pauli, the $K_s$ can now act on the logical state as non-Clifford operators. On the bit-flip code with $U=\bigotimes_je^{i\theta X_j}$, each $K_s$ is proportional to a logical $X$ rotation: $e^{-i\phi X_L}$ with $\tan\phi=\tan^3\theta$ for the trivial syndrome, which has probability $\cos^6\theta+\sin^6\theta$, and $e^{i\theta X_L}$ for each weight-1 syndrome, as found by Greenbaum and Dutton for this code~\cite{Greenbaum2018} and by Huang \emph{et al} for the equivalent repetition code with $Z$ rotations~\cite{Huang2019}. For $\rho_L=\ket0$ and $0\le\theta\le\pi/4$ the robustness after extraction is $\cos2\theta\,(1+\tfrac12\sin^22\theta)+\sin^32\theta$, which is $7\sqrt2/8$ at $\theta=\pi/8$: after extraction the logical state holds magic it did not have before the error. For $\rho_L=\ket{H}$ the $X_L$ rotations leave $r_x=1/\sqrt2$ unchanged and rotate $(r_y,r_z)$, of length $1/\sqrt2$, so the robustness after extraction is $[1+\cos2\theta\,(1+\tfrac12\sin^22\theta)+\sin^32\theta]/\sqrt2$, which is $7/8+1/\sqrt2$ at $\theta=\pi/8$. On the $[[4,2,2]]$ code with $U=\bigotimes_je^{i\theta Z_j}$, only the trivial syndrome and the syndrome of odd-weight $Z$ errors occur, with $K=\mathds{1}-(1-\cos4\theta)\proj{00}_L$ and $K=i\sin4\theta\,\proj{00}_L$, where $\ket{00}_L$ is the $+1$ eigenstate of both logical $Z$ operators. Syndrome extraction is then a measurement of $\proj{00}_L$, weak at small $\theta$ and projective at $\theta=\pi/8$. For $\rho_L=\ket{H}\otimes\ket{0}$ the robustness after extraction is $1+\sqrt2|\cos4\theta|-\cos^24\theta$. For $0<\theta\le\pi/8$ it exceeds $\RoM(\rho_L)=\sqrt2$ when $\theta<\frac14\arccos(\sqrt2-1)$, reaches its maximum $3/2$ at $\theta=\pi/16$, and falls to 1 at $\theta=\pi/8$. At $\theta=\pi/16$ the state after the error is Clifford-equivalent to the CS resource state $\mathrm{CS}\ket{++}$ with two appended $\ket{0}$, of robustness $11/5$~\cite{HowardCampbell2017}. At $\theta=\pi/8$ it is Clifford-equivalent to $\ket{+}\otimes T\ket{+}$, again with two appended $\ket{0}$, of robustness $\sqrt2$, and extraction removes all of its magic. \Cref{tab:global} lists these cases. Under \eqref{eq:C} syndrome extraction returns $\RoM(\rho_L)$, whereas after a uniform rotation the decoded logical state can hold more or less magic than $\rho_L$.

\begin{table*}[t]
\caption{Uniform rotations. Robustness of the state before syndrome extraction, $\RoM(U\renc U^\dagger)$, and after it, $\RoM(\EC(U\renc U^\dagger))=\sum_sP(s)\RoM(\rho_L(s))$ by \cref{lem:sector}. The values after extraction are $7\sqrt2/8$, $7/8+1/\sqrt2$, $3/2$ and $1$. Before extraction, 2.218951 is $\RoM(\ket{H}^{\otimes3})$~\cite{HowardCampbell2017} and the $[[4,2,2]]$ values are $11/5$, the robustness of $\mathrm{CS}\ket{++}$, and $\sqrt2$. Values before extraction are from linear programming over all stabiliser states, with $11/5$ and $\sqrt2$ confirmed by Clifford equivalences and a linear program in rational arithmetic; the values after extraction follow from the logical Kraus operators above (accompanying notebook).}
\label{tab:global}
\begin{ruledtabular}
\begin{tabular}{llcccc}
code & $U$ & $\rho_L$ & $\RoM(\rho_L)$ & $\RoM(U\renc U^\dagger)$ & $\RoM(\EC(U\renc U^\dagger))$\\
\hline
bit-flip & $\bigotimes_je^{i\pi X_j/8}$ & $\ket{0}$ & 1 & 2.218951 & 1.237437\\
bit-flip & $\bigotimes_je^{i\pi X_j/8}$ & $\ket{H}$ & 1.414214 & 2.835330 & 1.582107\\
$[[4,2,2]]$ & $\bigotimes_je^{i\pi Z_j/16}$ & $\ket{H}\otimes\ket{0}$ & 1.414214 & 2.200000 & 1.500000\\
$[[4,2,2]]$ & $\bigotimes_je^{i\pi Z_j/8}$ & $\ket{H}\otimes\ket{0}$ & 1.414214 & 1.414214 & 1.000000\\
\end{tabular}
\end{ruledtabular}
\end{table*}

\emph{Clifford errors.} A Clifford error on every qubit can make syndrome extraction act as a measurement of a logical operator, which can remove robustness from the logical state itself. Let $U$ be a Clifford on all $n$ qubits, with logical Kraus operators $K_s=V_s^\dagger UV_0$ as above.

\begin{proposition}
\label{prop:measure}
Every $K_s$ is proportional to a unitary if and only if no element of $U^\dagger\Sg U$ lies in $\Nor$ outside $\pm\Sg$. Otherwise every non-zero $K_s$ is proportional to a proper projector on the logical space times a unitary.
\end{proposition}

\begin{proof}
$K_s^\dagger K_s=V_0^\dagger\Pi_0'V_0$ with $\Pi_0'$ the projector onto the code of $U^\dagger E_s\Sg E_sU$, which is $U^\dagger\Sg U$ up to signs. Hence $K_s^\dagger K_s=|\Sg|^{-1}\sum_{h\in U^\dagger\Sg U}\pm V_0^\dagger hV_0$, and $V_0^\dagger hV_0$ is $\pm\mathds{1}$ for $h\in\pm\Sg$, a signed logical Pauli for $h\in\Nor\setminus\pm\Sg$, and zero otherwise. The elements of $U^\dagger\Sg U$ in $\Nor$ form a group whose logical images commute, so the sum is a multiple of a projector. If that image is trivial modulo signs, the projector is the identity or zero. Otherwise it is zero or a proper projector, and it is non-zero for some $s$ since $\sum_sK_s^\dagger K_s=\mathds{1}$.
\end{proof}

Conjugating by $U$, syndrome extraction after $U$ measures the generators of $U^\dagger\Sg U$ on the code state. \Cref{prop:measure} is therefore the rule for code deformation, stated for the Kraus operators $K_s$: measuring a new stabiliser group measures those of its elements that are logical operators of the old code~\cite{Vuillot2019}.

Three forms of this criterion are useful.

\emph{Products of single-qubit Cliffords} preserve supports. Some product of single-qubit Cliffords makes syndrome extraction measure a logical operator if and only if some stabiliser and some non-trivial logical operator have the same support. There is no such pair on the bit-flip, $[[5,1,3]]$, Steane and $[[9,1,3]]$ surface codes, while on $[[4,2,2]]$ 18 non-trivial logical operators, such as $ZXZX$, share their support with a stabiliser ($XXXX$). For example, Hadamard gates on qubits 3 and 4 of the $[[4,2,2]]$ code take $XXXX$ to $X_1X_2Z_3Z_4$, a logical operator, so syndrome extraction measures it. Since $Z_3Z_4=Z_1Z_2$ on the code space, it is the product of the logical $X$ of the first logical qubit, $X_1X_2$, and the logical $Z$ of the second, $Z_1Z_2$. For $\rho_L=\ket{H}\otimes\ket{0}$ the second logical qubit is in the $+1$ eigenstate of that $Z$, so syndrome extraction acts as a measurement of the logical $X$ of the first. The error is Clifford, so the robustness of the encoded state stays $\sqrt2$, but every outcome leaves a stabiliser state, so syndrome extraction lowers it to 1. Unlike under \eqref{eq:C}, the removed magic is magic of the logical state.

\emph{Commuting $\pi/4$ Pauli rotations} $U=\prod_j\exp(i\pi P_j/4)$ give $U^\dagger gU=g\,Q_g$ up to a phase for $g\in\Sg$, with $Q_g$ the product of the rotation axes that anticommute with $g$. Syndrome extraction measures a logical operator if and only if, for some stabiliser $g$, $Q_g$ is a non-trivial logical operator.

For $Z$ rotations by any angle on the repetition code with stabilisers $X_jX_{j+1}$, Huang \emph{et al}\ found that every $K_s$ is proportional to a unitary for odd $n$, as Greenbaum and Dutton had found earlier~\cite{Greenbaum2018}, while for even $n$ the $K_s$ measure the logical operator $Z^{\otimes n}$, weakly at small angles, and project onto its eigenstates for the syndromes of errors of weight $n/2$~\cite{Huang2019}. At angle $\pi/4$ these rotations are Clifford, and the second form gives the same parity dependence: a stabiliser $X^B$ gives $Q_g=Z^B$, which is a non-trivial logical operator only when $B$ is the whole code, and $X^{\otimes n}$ is a stabiliser only for even $n$.

\emph{General Cliffords} reduce to linear algebra over $\Fq$. Write Paulis as vectors $(x|z)\in\Fq^{2n}$ and let $M$ be the symplectic matrix of $P\mapsto U^\dagger PU$.

\begin{proposition}
\label{prop:F2}
For $g\in\Sg$, $U^\dagger gU$ is $g\,\eta_g$ up to a phase, with $\eta_g=(M+\mathds{1})g$ linear in $g$. Syndrome extraction after $U$ measures a logical operator if and only if some combination $\sum_ic_i\eta_{g_i}$ of the generators' $\eta$'s commutes with every generator without lying in the span of $\Sg$. For a CSS code with check spaces $C_X$, $C_Z\subseteq\Fq^n$:
\begin{enumerate}[label=(\roman*),leftmargin=*,itemsep=0pt]
\item a CNOT circuit with $UX^aU^\dagger=X^{Ga}$ measures a logical operator iff $(G^{-1}+\mathds{1})C_X$ contains an $X$-logical support (a vector of $C_Z^\perp$ outside $C_X$) or $(G^{T}+\mathds{1})C_Z$ contains a $Z$-logical support,
\item a circuit of CZ and S gates, with $Q$ the symmetric matrix holding the CZ graph off the diagonal and the S pattern on it, measures a logical operator iff $QC_X$ meets $C_X^\perp$ outside $C_Z$.
\end{enumerate}
The verdict is the same for $U$ and $U^\dagger$.
\end{proposition}

\begin{proof}
$U^\dagger gU$ has vector $Mg=g+\eta_g$, and membership in $\Nor$ and in $\Sg$ are conditions on vectors, so phases do not matter. The elements of $U^\dagger\Sg U$ correspond to combinations of the generators. For CSS codes $X^aZ^b\in\Nor$ if and only if $a\in C_Z^\perp$ and $b\in C_X^\perp$, so the condition splits by type. For CNOT circuits $U^\dagger X^aU=X^{G^{-1}a}$ and $U^\dagger Z^bU=Z^{G^Tb}$, and for CZ and S circuits $U^\dagger X^aU\propto X^aZ^{Qa}$ while $Z^b$ is unchanged. For the last statement let $\Sigma\subseteq\Fq^{2n}$ be the stabiliser space and $\Sigma'=M\Sigma$, and write $\perp$ for the symplectic complement. The condition is $\Sigma'\cap\Sigma^\perp\subseteq\Sigma$, equivalently $\dim(\Sigma'\cap\Sigma^\perp)=\dim(\Sigma'\cap\Sigma)$. Since $\dim\Sigma'=\dim\Sigma$, $\dim(\Sigma'\cap\Sigma^\perp)=\dim(\Sigma\cap\Sigma'^\perp)$. Applying the symplectic map $M$ to the condition for $M^{-1}$, $\dim(M^{-1}\Sigma\cap\Sigma^\perp)=\dim(M^{-1}\Sigma\cap\Sigma)$, turns it into $\dim(\Sigma\cap\Sigma'^\perp)=\dim(\Sigma\cap\Sigma')$, which is the same condition.
\end{proof}

\begin{corollary}
\label{cor:parity}
On a CSS code whose checks have even weight and whose non-trivial $X$-type and $Z$-type logical operators have odd weight, such as odd-distance rotated surface codes and the Steane code, no layer of commuting $\pi/4$ rotations about even-weight $X$-type and $Z$-type axes measures a logical operator.
\end{corollary}

\begin{proof}
For $g,g'\in\Sg$, $Q_{gg'}=Q_gQ_{g'}$ up to a phase. Every rotation axis has even weight, so every $Q_g$ has $X$ and $Z$ parts of even weight. An element of $\Nor$ whose $X$ and $Z$ parts have even weight lies in $\Sg$, so no $Q_g$ is a non-trivial logical operator.
\end{proof}

By \cref{cor:parity}, on rotated surface codes of every odd distance, layers of $\pi/4$ rotations about $ZZ$ or $XX$ on all nearest-neighbour edges (for $ZZ$ also with the diagonal edges) do not measure a logical operator. At $d=3,5,7$ the same holds for CZ on all nearest-neighbour edges, with or without a following layer of $H$ gates, and for $S$ or $H$ on every qubit. Every horizontal or vertical CNOT circuit we tried measures one. The row cascades $\mathrm{CNOT}(c,c+1)$, $c=0,\dots,d-2$, applied from either end, the layers of disjoint pairs with $c$ even or with $c$ odd, and the even layer followed by the odd one each measure the logical $X$ operator, an $X$ string on a column. The vertical cascade and the disjoint vertical pairs measure the logical $Z$ operator, a $Z$ string on a row. The mechanism is visible in \cref{prop:F2}(i): under a layer of disjoint horizontal CNOTs, $\eta_g$ of an $X$ check is $X$ on the targets of its control qubits, a vertical pair for a plaquette, a single qubit for a boundary check, or nothing for a check with no control qubit, and these combine along a target column into the column's $X$ string (\cref{fig:surface}). Partial CZ layers behave differently from full ones: on $d=3$, with qubits numbered row by row from 0, the two gates $\mathrm{CZ}(0,1)\mathrm{CZ}(1,2)$ on the first row already measure the logical $Z$ operator, and of the 4095 non-empty subsets of the twelve nearest-neighbour CZ gates, applied as a layer, 1840 measure a logical operator, while by \cref{cor:parity} no subset of $\pi/4$ $ZZ$ or $XX$ rotations does. A CZ is a $ZZ$ rotation times $S$ gates, and the $S$ gates break the parity.

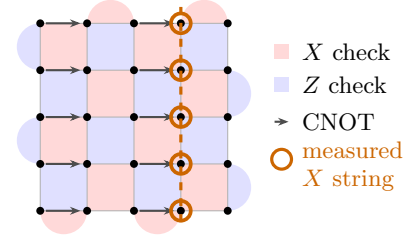
\begin{figure}[t]
\centering
\begin{tikzpicture}[scale=0.62]
  \foreach \r in {0,...,3} \foreach \c in {0,...,3} {
    \pgfmathtruncatemacro{\t}{mod(\r+\c,2)}
    \ifnum\t=0 \fill[red!15] (\c,-\r) rectangle (\c+1,-\r-1); \else \fill[blue!12] (\c,-\r) rectangle (\c+1,-\r-1); \fi
  }
  \foreach \c in {1,3} { \fill[red!15] (\c,0) arc (180:0:0.5) -- cycle; }
  \foreach \c in {0,2} { \fill[red!15] (\c,-4) arc (-180:0:0.5) -- cycle; }
  \foreach \r in {1,3} { \fill[blue!12] (4,-\r) arc (90:-90:0.5) -- cycle; }
  \foreach \r in {0,2} { \fill[blue!12] (0,-\r) arc (90:270:0.5) -- cycle; }
  \draw[gray!60] (0,0) grid[step=1] (4,-4);
  \foreach \r in {0,...,4} \foreach \c in {0,2} {
    \draw[-{Stealth[length=4pt]},thick,black!70] (\c+0.12,-\r) -- (\c+0.86,-\r);
  }
  \foreach \r in {0,...,4} \foreach \c in {0,...,4} { \fill[black] (\c,-\r) circle (2.4pt); }
  \foreach \r in {0,...,4} { \draw[orange!80!black,line width=1.4pt] (3,-\r) circle (6pt); }
  \draw[orange!80!black,line width=1.2pt,dashed] (3,0.35) -- (3,-4.35);
  \fill[red!15] (5.0,-0.75) rectangle (5.3,-0.45);
  \node[anchor=west,font=\small] at (5.4,-0.6) {$X$ check};
  \fill[blue!12] (5.0,-1.45) rectangle (5.3,-1.15);
  \node[anchor=west,font=\small] at (5.4,-1.3) {$Z$ check};
  \draw[-{Stealth[length=4pt]},thick,black!70] (5.0,-2.1) -- (5.3,-2.1);
  \node[anchor=west,font=\small] at (5.4,-2.1) {CNOT};
  \draw[orange!80!black,line width=1.2pt] (5.15,-2.9) circle (6pt);
  \node[anchor=west,font=\small,align=left,orange!80!black] at (5.4,-3.05) {measured\\ $X$ string};
\end{tikzpicture}
\caption{A layer of disjoint horizontal CNOTs (arrows, control to target) on the $d=5$ rotated surface code, $X$ checks red, $Z$ checks blue. Each $X$ check acquires $X$ on the targets of its controls, and these combine into the $X$ string on one column (circled), which commutes with every check and is not a stabiliser: syndrome extraction after this layer measures a logical operator.}
\label{fig:surface}
\end{figure}

\section{Magic-state memories}
\label{sec:memories}

We ask how many rounds a stored single-qubit magic state stays outside the stabiliser polytope under Pauli noise and, at the end of the section, under a uniform rotation repeated in every round. With syndrome extraction and Pauli recovery in every round, \cref{lem:sector} holds round after round, and the one-qubit formula for $\RoM$ gives the robustness in closed form and the lifetime as the first round at which it reaches 1. Let the decoded logical channel have class probabilities $q_I,q_X,q_Y,q_Z$ and eigenvalues
\begin{equation}
\begin{aligned}
\Lambda_x&=1-2(q_Y+q_Z),\quad\Lambda_y=1-2(q_X+q_Z),\\
\Lambda_z&=1-2(q_X+q_Y).
\end{aligned}
\end{equation}
Each $q_L$ is the total probability of the errors that the decoder leaves as the logical Pauli $L$, a sum over syndrome sectors of coset probabilities. After $m$ rounds with the syndrome record discarded, a logical state with Bloch vector $\mathbf{r}$ has Bloch vector $(\Lambda_x^mr_x,\Lambda_y^mr_y,\Lambda_z^mr_z)$ and $\RoM=\max(1,\sum_k|\Lambda_k|^m|r_k|)$, and its magic lifetime, the number of rounds until it enters the stabiliser polytope, is
\begin{equation}
m^*=\min\Big\{m:\ \sum_k|\Lambda_k|^m|r_k|\le1\Big\}.
\label{eq:lifetime}
\end{equation}
Keeping the record gives the syndrome-resolved value $\sum_sP(s)\RoM(\rho_L(s))$ of \cref{lem:sector}, which is at least as large. The difference is robustness stored in the classical record.

By \cref{eq:lifetime} the lifetime depends on the orientation of the stored state. The $y$ component is flipped by both $X_L$ and $Z_L$ errors, and $\Lambda_y$ is the smallest eigenvalue if and only if $q_Y\le\min(q_X,q_Z)$. On a CSS code with equal $X$ and $Z$ distances in which no minimum-weight $X$ logical operator has the same support as a minimum-weight $Z$ logical operator, every logical $Y$ operator is heavier than the distance. On the rotated surface code the lightest logical $Y$ operators have weight $2d-1$ at $d=3$ and $5$, so a residual logical $Y$ needs about $d$ errors against $(d+1)/2$ for $X_L$ or $Z_L$, and at low noise with physical $X$, $Y$ and $Z$ rates of the same order, such as depolarising noise, $q_Y\le\min(q_X,q_Z)$. A state in the $xz$ plane, such as the Hadamard eigenstate with Bloch vector $(1,0,1)/\sqrt2$~\cite{BravyiKitaev2005}, then outlives $\ket{H}$, whose Bloch vector $(1,1,0)/\sqrt2$ lies in the $xy$ plane. With $q_X=q_Z=q\ll1$ and $q_Y$ negligible, the $xz$ lifetime is, to leading order, $\ln\sqrt2/(2q)$ and the $xy$ lifetime $\ln(1/x_*)/(2q)$, with $x_*=(\sqrt{1+4\sqrt2}-1)/2\approx0.790$ the positive root of $x+x^2=\sqrt2$, a ratio of $1.47$ that is independent of $q$ and so of the distance. \Cref{tab:lifetimes} gives code-capacity lifetimes. The rotated $[[9,1,3]]$ surface code shows the factor (265 against 182 rounds). On the Steane code, where $X_L$, $Y_L$ and $Z_L$ all have weight-3 representatives on the same support, the code prefers no plane. There the decoder ties between three logical classes in 42 of the 64 sectors, and the tie-breaking rule decides which plane, if any, is preferred. Over 400 random rules the median lifetime in each plane is 165 rounds. Breaking every tie so that $Z_L$ is the least likely residual logical error gives 192 rounds in the $xy$ plane and 155 in the other two. The corresponding rules for $X_L$ and $Y_L$ favour the $yz$ and $xz$ planes, respectively. The quantum Reed--Muller $[[15,1,3]]$ code, with $Z$-distance 3 and $X$-distance 7, protects the $z$ component almost perfectly, and there states in the planes containing $z$ outlive those in the $xy$ plane by a factor 2.5.

\emph{Coherent noise.} If the same uniform rotation $\bigotimes_je^{i\theta X_j}$ acts on the bit-flip code in every round, the logical Kraus operators of \cref{sec:clifford} are commuting logical $X$ rotations. With the record discarded, the decoded logical channel keeps $r_x$ and multiplies $w=r_y+ir_z$ by
\begin{equation}
\lambda=(c^3+is^3)^2+3c^2s^2(c-is)^2
\label{eq:lambda}
\end{equation}
in each round, with $c=\cos\theta$ and $s=\sin\theta$, so after $m$ rounds $\RoM=\max(1,|r_x|+|\mathrm{Re}\,\lambda^mw|+|\mathrm{Im}\,\lambda^mw|)$. The Pauli-twirled noise, in which each qubit suffers an $X$ error with probability $s^2$, multiplies $w$ by $\mathrm{Re}\,\lambda=1-2q_X$, with $q_X=3c^2s^4+s^6$, and for small $\theta$ both $|\lambda|$ and $\mathrm{Re}\,\lambda$ are $1-6\theta^4+O(\theta^6)$. The phase, $|\arg\lambda|=4\theta^3+O(\theta^5)$, rotates $w$. Since $|\mathrm{Re}\,z|+|\mathrm{Im}\,z|\ge|z|$, the rotation can only raise $|\mathrm{Re}\,\lambda^mw|+|\mathrm{Im}\,\lambda^mw|$ above $|\lambda|^m|w|$, and the robustness is no longer monotone in $m$: a state can leave the stabiliser polytope again after entering it. For each state we give two rounds: its magic lifetime, the first round at which it enters the polytope, and the round from which it stays inside. At $\theta=0.1$, $\ket{H}$ and the $xz$-plane state, whose $w$ starts on an axis, enter after 1589 rounds, against 1489 under the twirled noise and the leading-order value $\ln(1+\sqrt2)/(6\theta^4)=1469$. This is the first round after $|\lambda|^m\le\sqrt2-1$ at which $w$ is near an axis, which it passes about every $\pi/(8\theta^3)$ rounds, and the states stay inside from round 1918. The $yz$-plane state, whose $w$ starts midway between the axes, would reach an axis after about $\pi/(16\theta^3)$ rounds, but enters earlier, after 172 rounds, against 586. Its robustness rises again to 1.13 after 363 rounds, and it stays inside from round 489. The face state enters after 1361 rounds, against 1698, and stays inside from round 1666, so under this noise it is no longer the longest-lived. A stored $\ket0$, a stabiliser state, gains magic: its robustness exceeds 1 from the first round and reaches its maximum, 1.27, after 162 rounds.

\begin{table}[t]
\caption{Lifetimes $m^*$ of \cref{eq:lifetime} (rounds, syndrome record discarded) of single-qubit magic states in a code-capacity model: on each qubit in each round $X$, $Y$ and $Z$ errors each occur with probability $p/3$, $p=0.01$, followed by perfect syndrome extraction and a coset maximum-likelihood decoder~\cite{BravyiSucharaVargo2014}, which takes each $E_s$ from the most probable coset of $\Sg$ with syndrome $s$. The states Clifford-equivalent to $\ket{H}$ used here have Bloch vectors $(1,1,0)/\sqrt2$ ($xy$), $(1,0,1)/\sqrt2$ ($xz$) and $(0,1,1)/\sqrt2$ ($yz$), and the face state $(1,1,1)/\sqrt3$. For the Steane code the table gives the median over 400 random rules for breaking ties between logical classes. For the other codes the lifetimes are the same for every tie-breaking rule we tried.}
\label{tab:lifetimes}
\begin{ruledtabular}
\begin{tabular}{lcccc}
code & $xy$ & $xz$ & $yz$ & face\\
\hline
unencoded & 26 & 26 & 26 & 41\\
Steane $[[7,1,3]]$ & 165 & 165 & 165 & 261\\
surface $[[9,1,3]]$ & 182 & 265 & 182 & 328\\
Reed--Muller $[[15,1,3]]$ & 54 & 136 & 136 & 156\\
\end{tabular}
\end{ruledtabular}
\end{table}

\section{Odd prime dimension}
\label{sec:odd}

In odd prime dimension the negative volume of the Wigner function, $N(\rho)=\sum_u|W_\rho(u)|-1$, twice the sum negativity of Ref.~\cite{Veitch2014}, is a second measure of magic, and lifetimes defined by it and by robustness can be compared. \Cref{lem:sector} holds in every dimension, and $N$ obeys the same sector rule: the negative volume $N_C$ of the code-adapted representation of Ref.~\cite{BozzettoHance2026} satisfies $N_C(\rho)=\sum_sP(s)N(\rho_L(s))$, and $N+1\le\RoM$ because stabiliser states have non-negative Wigner functions~\cite{Gross2006}. \Cref{thm:coherent} also holds for qudit stabiliser codes of odd prime dimension, with Weyl operators in place of Paulis (and $gOg^\dagger$, $P\renc P^\dagger$ in place of $gOg$, $P\renc P$ in \cref{lem:twirl}), $\ket{\Phi}$ a product of maximally entangled qudit pairs and $\RoM$ defined by the qudit stabiliser states: the proofs of \cref{lem:twirl,lem:bell} use only that the Weyl operators modulo phases form a symplectic space over $\mathbb{Z}_d$ on which Clifford unitaries realise every symplectic map~\cite{Gross2006}, and the four properties used in the proof of \cref{thm:append} hold for qudits. Since $N+1=\sum_u|W_\rho(u)|$ is multiplicative~\cite{Veitch2014} and unchanged by Clifford isometries~\cite{Gross2006}, the analogue of \cref{eq:main} for the negative volume is $N(U\renc U^\dagger)+1=(N(J_U)+1)(N(\rho_L)+1)$ and $N_C(U\renc U^\dagger)=N(\rho_L)$. The mana $\log(N+1)$~\cite{Veitch2014} of the state after the error therefore exceeds that of $\rho_L$ by the mana of $J_U$, for every logical state.

The negative volume does not determine the robustness: of the two kinds of qutrit state of largest sum negativity~\cite{Veitch2014}, the strange state $(\ket1-\ket2)/\sqrt2$ has $N=2/3$ and $\RoM=3$, and the Norrell state $(-\ket0+2\ket1-\ket2)/\sqrt6$ has $N=2/3$ and $\RoM=2$. Lifetimes defined by the two measures agree when the noise never takes the state through a Wigner-positive non-stabiliser state. Such states exist and are bound magic states~\cite{Veitch2012,Veitch2014}. The following theorem gives, in each odd prime dimension, a state and a class of channels for which the lifetimes agree.

Let $A_0$ be the parity operator and $\Pi_-$ its $-1$ eigenspace, and let $\rho_-=\Pi_-/\Tr\Pi_-$. This is the unique state invariant under the symplectic part of the Clifford group with $W(0)=-1/d$, since the invariant operators are spanned by $\mathds{1}$ and $A_0$~\cite{Gross2006}. Its Wigner function is $-1/d$ at the origin and $1/(d(d-1))$ elsewhere~\cite{Gross2006}, and at $d=3$ it is the strange state. Call a Weyl channel $\Lambda_p(\rho)=\sum_vp_vD_v\rho D_v^\dagger$, with $D_{(a,b)}=X^aZ^b$, radial if $p_{\lambda v}=p_v$ for every non-zero $\lambda\in\mathbb{Z}_d$. Depolarising noise and complete dephasing along any Weyl axis are radial, and at $d=3$ radial means symmetric, $p_v=p_{-v}$.

\begin{theorem}
\label{thm:radial}
For every odd prime $d$ and every radial $p$, $\Lambda_p(\rho_-)$ is a stabiliser mixture if and only if its Wigner function is non-negative, if and only if $\max_vp_v\le1/d$.
\end{theorem}

\begin{proof}
A Weyl channel convolves the Wigner function with $p$, so $W(\Lambda_p\rho_-)(u)=(1/d-p_u)/(d-1)$, and the claim is that $h=1-dp$ lies in the cone spanned by the indicator functions of the $d(d+1)$ lines of $\mathbb{Z}_d^2$, the Wigner functions of pure stabiliser states~\cite{Gross2006}, whenever $0\le h\le1$. Averaging over the scalings $v\mapsto\lambda v$, radial generators suffice: the line through the origin in direction $k$, $(1;\mathbf{e}_k)$, and the union of the other $d-1$ lines of direction $k$, $(0;\mathbf{1}-\mathbf{e}_k)$, in coordinates listing the value at the origin and on the punctured line of each of the $d+1$ directions. Writing $h=\sum_ka_k(1;\mathbf{e}_k)+\sum_kb_k(0;\mathbf{1}-\mathbf{e}_k)$ forces $B=\sum_kb_k=(\sum_jh_j-h_0)/d$ and $a_j=h_j-B+b_j$, so a non-negative decomposition exists if and only if $\sum_j\max(0,B-h_j)\le B$. With $\sum_uh_u=d(d-1)$, $\epsilon=1-B=h_0/(d-1)$ and $\sum_jh_j=d-\epsilon$, this is $\sum_j\max(0,h_j-B)\le(d-1)\epsilon$. Since $h_j\le1$, each term is at most $\epsilon$. If at most $d-1$ of the $h_j$ exceed $B$ the sum is at most $(d-1)\epsilon$, and if at least $d$ do, then, since $h_j\ge0$ and $B=1-h_0/(d-1)\ge0$, it is at most $\sum_jh_j-dB=(d-1)\epsilon$. The converse holds because stabiliser mixtures have non-negative Wigner functions.
\end{proof}

For depolarising noise the theorem has a short proof. At the Wigner threshold the state has $W(0)=0$ and $W(u)=1/(d^2-1)$ elsewhere, which is the Wigner function of the uniform mixture of the $d^2-1$ pure stabiliser states whose lines miss the origin, and stronger noise mixes this state with $\mathds{1}/d$.

A composition of radial channels is again radial, so under radial noise the negativity-based and robustness-based lifetimes of $\rho_-$ coincide. Neither the state nor the channel class can be relaxed in general (Appendix~\ref{app:odd}). At $d=3$ the Wigner-positive non-stabiliser state of Ref.~\cite{Gross2006}, an equal mixture of the strange state and two of its displacements, is $\Lambda_p(\rho_-)$ for a Weyl channel that is uniform on three displacements, one of them the identity. This channel has $\max_vp_v=1/3$ but is not symmetric.

\section{Discussion}
\label{sec:discussion}

Under Pauli noise the robustness of an encoded magic state is the syndrome-weighted robustness of its logical states. After a coherent error on a correctable region that supports no stabiliser, it is the robustness of the logical state appended with the error's Choi state. The excess over the logical robustness is strictly positive whenever the error is not Clifford, and syndrome extraction removes this excess, leaving the logical robustness unchanged. Errors on regions that can hold a logical operator are (apart from rotations about Paulis that anticommute with a stabiliser) controlled by the magic capacity, which we have shown equals the Choi-state robustness for every single-qubit gate and for every two-qubit gate Clifford-equivalent to a diagonal or single-qubit gate. A Clifford error on every qubit preserves the robustness of every logical state unless $U^\dagger gU$ is a non-trivial logical operator for some stabiliser $g$, in which case syndrome extraction measures a logical operator. This condition reduces to linear algebra over $\Fq$. Note we have treated syndrome extraction as ideal, and did not model faults inside the extraction circuit, which can spread one fault over part of a check's support.

Several open questions remain. The question of whether $\capC(U)=\RoM(J_U)$ is open for two-qubit gates that are Clifford-equivalent neither to a diagonal nor to a single-qubit gate, and for diagonal gates on three and four qubits (our numerical checks support this in both cases, but we don't yet have an analytical proof). For three-qubit diagonal gates the subspace inputs include two-qubit phase states, whose robustness cannot always be certified by anticommuting witnesses, so the argument in the proof of \cref{thm:capacity2} (Appendix~\ref{app:capacity2}) does not extend directly to three qubits. For the three states $\sigma$ of Appendix~\ref{app:tightness}, \cref{thm:append} is attained by some $\rho$ for every value of $\RoM(\rho)$ with $1.6\le\RoM(\rho)\le\sqrt3$, and for $\sigma=H$ also with $\sqrt2\le\RoM(\rho)<1.6$. Elsewhere, and for other $\sigma$, it is not known whether the bound is attained. Unitary errors that act between two rounds of syndrome extraction combine into one unitary on the union of their regions, so \cref{thm:coherent} applies to the product when the union satisfies \eqref{eq:C}, \cref{cor:stabA} when the erasure of the union is correctable, and \cref{cor:pauli} when the product is a rotation about a Pauli outside $\Nor$. In other cases, as for the rotations of qubits 1 and 2 of the $[[4,2,2]]$ code in \cref{sec:coherent}, the robustness before extraction is known only in examples and is bounded by the magic capacity. On non-CSS codes the criteria of \cref{sec:clifford} do not split into $X$-type and $Z$-type conditions. For memories, the orientation rule of \cref{sec:memories} is a choice of logical frame at injection, and the size of its gain at circuit level depends on the syndrome-extraction schedule through the relative rates of logical $X$ and $Z$ errors. Under a uniform rotation repeated in every round of the bit-flip code, the logical rotation moves the stored state between orientations, a state can regain magic after entering the stabiliser polytope, and the face state is no longer the longest-lived (\cref{sec:memories}).

\textit{Acknowledgements -} JRH acknowledges support from a Royal Society Research Grant (RG/R1/251590), an EPSRC Mathematical Sciences Small Grant (UKRI3647), and from their EPSRC Quantum Technologies Career Acceleration Fellowship (UKRI1217).

\textit{Data availability -} A Wolfram Mathematica notebook reproducing all numerical results in this paper is openly available on Zenodo~\cite{BozzettoHance2026Notebook}.

\bibliographystyle{unsrturl}
\bibliography{ref.bib}

\appendix

\section{Proofs of the twirl lemma and the Bell-pair normal form}
\label{app:bell}

\begin{proof}[Proof of \cref{lem:twirl}]
Expanding $\Pi_s=|\Sg|^{-1}\sum_g\chi_s(g)g$ in terms of the characters $\chi_s$ of $\Sg$ and using their orthogonality gives the first statement. Since $g\renc g=\renc$ and $U$ is supported on $A$, $gU\renc U^\dagger g=(g_AUg_A)\renc(g_AUg_A)^\dagger$, with $g_A$ the restriction of $g$ to $A$. The restriction map is a homomorphism whose image $G_A$ has commutant $\Nor\cap\Pa_A$ in $\Pa_A$, and each element of $G_A$ has the same number of preimages, so the average over $\Sg$ is the average over $G_A$. Under \eqref{eq:C}, $G_A=\Pa_A$ and the average is the Pauli twirl. If the erasure is correctable but $A$ supports stabilisers, the commutant consists of stabilisers, the average keeps the products $u_Pu_{P'}^*P\renc P'$ with $PP'$ in the commutant, and each such group of terms is a multiple of $P\renc P$ because stabilisers act trivially on $\renc$. Paulis in different groups differ by an element outside $\Nor$ and so have different syndromes.
\end{proof}

\begin{proof}[Proof of \cref{lem:bell}]
Under \eqref{eq:C} the restriction map $\Sg\to\Pa_A$ is onto modulo phases (proof of \cref{lem:twirl}), so for each $j\in A$ there are $g_{X,j},g_{Z,j}\in\Sg$ whose restrictions to $A$ are $X_j$ and $Z_j$. Their parts $a_{X,j}$, $a_{Z,j}$ outside $A$ inherit the commutation relations of the $X_j$ and $Z_j$: the $g$'s commute, so two outside parts anticommute if and only if the corresponding parts on $A$ do. None of the $2|A|$ outside parts is a phase, since then $X_j$ or $Z_j$ itself would lie in $\Sg$, against \eqref{eq:C}. They are also independent, since a product of a non-empty subset equal to a phase would commute with everything, while the partner of any member anticommutes with it. Condition \eqref{eq:C} forces $4^{|A|}\le2^{n-k}$, so the complement of $A$ has room, and a Clifford $W$ on the complement sends $a_{X,j}$ and $a_{Z,j}$ to $X_{j'}$ and $Z_{j'}$ on fresh qubits $A'$. The conjugated group then contains $\pm X_jX_{j'}$ and $\pm Z_jZ_{j'}$ for every $j$, and a Pauli on $A'$ absorbed into $W$ removes the signs, so the pairs are in $\ket{\Phi^+}$. Every other generator, and every logical operator, commutes with these and can be multiplied by them to act trivially on $A\cup A'$, since the commutant of $\langle XX,ZZ\rangle$ in the two-qubit Pauli group is $\langle XX,ZZ\rangle$. What remains is a stabiliser code $\Sg''$ on $n-2|A|$ qubits with the same logical qubits, and $WV_0$ is a Clifford isometry with image $\ket{\Phi^+}^{\otimes|A|}\otimes\mathcal{H}''$, with $\mathcal{H}''$ the code space of $\Sg''$, so the induced encoder $V_0''$ is a Clifford isometry.
\end{proof}

\section{Tightness of the appending bound}
\label{app:tightness}

For some $\sigma$ the bound of \cref{thm:append} is attained over a whole range of $r=\RoM(\rho)$. Let $\rho_r$ be a face state, with Bloch vector $(r/3)(\pm1,\pm1,\pm1)$, so that $\RoM(\rho_r)=r$. For $\sigma=H$ and for the Choi states of $\exp(i\pi X/16)$ and of the Hadamard-axis rotation of \cref{tab:gap}, each of the eight sign choices gives $\RoM(\sigma\otimes\rho_r)-r=(r+1)(\mu(\sigma)-1)$ at $r=1.6$ and at $r=\sqrt3$, and for $\sigma=H$ also at $r=\sqrt2$. Since $\rho_r$ is affine in $r$ and $\RoM$ is convex, equality then holds for $1.6\le r\le\sqrt3$, and for $\sigma=H$ for $\sqrt2\le r\le\sqrt3$. So for these $\sigma$ no bound that depends only on $\sigma$ and $\RoM(\rho)$ improves on \cref{thm:append} in these ranges. At smaller $r$ the face states need not attain the bound: for $\sigma=H$ and $r=1.1$ the smallest of the eight values is $0.3947$, while the larger of the two lower bounds of \cref{eq:gapbounds}, with $H$ in place of $J_U$, is $0.3142$. The proof of \cref{thm:append} discards the marginal constraint $\Tr_n\Gamma_+-\Tr_n\Gamma_-=\sigma$. In linear-programming tests with random single-qubit $\sigma$ next to $\ket0$, $\ket{H}$ and $\ket{F}$, keeping it gave nothing beyond the larger of the two lower bounds in \cref{eq:gapbounds}.

\section{Proof of the two-qubit capacity theorem}
\label{app:capacity2}

\begin{proof}[Proof of \cref{thm:capacity2}]
Both quantities are invariant under $U\to KUK'$, so the second case is \cref{prop:capacity1}, with the second qubit as part of the ancilla. For the first, let $U=\mathrm{diag}(e^{i\theta_x})$, $x\in\Fq^2$. Theorem~S1 in the supplementary material of Ref.~\cite{SeddonCampbell2019} reduces the capacity of a diagonal gate to $\capC(U)=\max_{\mathcal{A}}\RoM(U\ket{\mathcal{A}})$ over affine subspaces $\mathcal{A}\subseteq\Fq^2$, with $\ket{\mathcal{A}}$ the uniform superposition over $\mathcal{A}$, and Eq.~(S41) there gives $\RoM(J_U)=\RoM(U\ket{{+}{+}})$, the case $\mathcal{A}=\Fq^2$. A point gives $\RoM=1$. A line $L=\{x,x+b\}$ gives a state Clifford-equivalent to $\ket{0}\otimes(\ket{0}+e^{i\alpha_L}\ket{1})/\sqrt2$ with $\alpha_L=\theta_{x+b}-\theta_x$, of robustness $f(\alpha_L)=|\cos\alpha_L|+|\sin\alpha_L|$. It remains to show $\RoM(U\ket{{+}{+}})\ge f(\alpha_L)$ for each of the six lines. CNOT and $X$ gates act on $\Fq^2$ as the affine group, which is transitive on lines and maps $U\ket{{+}{+}}$ to $U'\ket{{+}{+}}$ with $U'$ diagonal, so we may take $L=\{00,01\}$. With $\alpha=\theta_{01}-\theta_{00}$, $\alpha'=\theta_{11}-\theta_{10}$, $\beta=\theta_{10}-\theta_{00}$, $\beta'=\theta_{11}-\theta_{01}$ and $\phi=\alpha'-\alpha=\beta'-\beta$, the state $\psi=U\ket{{+}{+}}$ has
\begin{align*}
\av{IX},\av{IY}&=\tfrac12(\cos\alpha+\cos\alpha'),\ \tfrac12(\sin\alpha+\sin\alpha'),\\
\av{ZX},\av{ZY}&=\tfrac12(\cos\alpha-\cos\alpha'),\ \tfrac12(\sin\alpha-\sin\alpha'),\\
\av{XI},\av{YI}&=\tfrac12(\cos\beta+\cos\beta'),\ \tfrac12(\sin\beta+\sin\beta'),\\
\av{XZ},\av{YZ}&=\tfrac12(\cos\beta-\cos\beta'),\ \tfrac12(\sin\beta-\sin\beta').
\end{align*}
For a set $M$ of pairwise anticommuting Paulis, $\Omega=\sum_{P\in M}\pm P$ satisfies $|\Tr(\Omega\sigma)|\le1$ on every stabiliser state, since a stabiliser state has a non-zero expectation on at most one member of $M$. By \cref{eq:dual}, $\RoM(\rho)\ge\sum_{P\in M}|\av{P}|$. For anticommuting $M$ the reduced stabiliser polytope of Ref.~\cite{Varela2026} is a cross-polytope by their Theorem~3, so this is their reduced robustness whenever it exceeds $1$. The sets $\{IX,IY,XZ,YZ\}$ and $\{XI,YI,ZX,ZY\}$ are pairwise anticommuting, so with $\mathbf{u}=(\cos\alpha,\sin\alpha)$, $\mathbf{u}'=(\cos\alpha',\sin\alpha')$, $\mathbf{w}=(\cos\beta,\sin\beta)$, $\mathbf{w}'=(\cos\beta',\sin\beta')$ and $\|\cdot\|_1$ the 1-norm on $\mathbb{R}^2$,
\begin{equation*}
\RoM(\psi)\ge B_\pm=\tfrac12\|\mathbf{u}\pm\mathbf{u}'\|_1+\tfrac12\|\mathbf{w}\mp\mathbf{w}'\|_1 .
\end{equation*}
Since $\mathbf{w}'$ is $\mathbf{w}$ rotated by $\phi$ and $f\ge1$, $\tfrac12\|\mathbf{w}+\mathbf{w}'\|_1\ge|\cos(\phi/2)|$ and $\tfrac12\|\mathbf{w}-\mathbf{w}'\|_1\ge|\sin(\phi/2)|$. Let $\boldsymbol{\varsigma}=(\mathrm{sgn}\cos\alpha,\mathrm{sgn}\sin\alpha)$ with $\mathrm{sgn}\,0=1$, so $\boldsymbol{\varsigma}\cdot\mathbf{u}=f(\alpha)$, and write $\alpha=\alpha_0+x$, $\alpha'=\alpha_0+y$ with $\alpha_0$ the direction of $\boldsymbol{\varsigma}$, an odd multiple of $\pi/4$. Then $\boldsymbol{\varsigma}\cdot\mathbf{u}=\sqrt2\cos x$, $\boldsymbol{\varsigma}\cdot\mathbf{u}'=\sqrt2\cos y$, $\phi=y-x$, and bounding each $\|\mathbf{u}\pm\mathbf{u}'\|_1$ below by its pairing with $\boldsymbol{\varsigma}$,
\begin{align*}
B_+-f(\alpha)&\ge|\sin\tfrac{y-x}{2}|\,\big(1-\sqrt2\,|\sin\tfrac{x+y}{2}|\big),\\
B_--f(\alpha)&\ge|\cos\tfrac{y-x}{2}|\,\big(1-\sqrt2\,|\cos\tfrac{x+y}{2}|\big).
\end{align*}
One of $|\sin\frac{x+y}{2}|$ and $|\cos\frac{x+y}{2}|$ is at most $1/\sqrt2$, so $\max(B_+,B_-)\ge f(\alpha)$.
\end{proof}

\section{Diagonal gates on more qubits}
\label{app:diagonal}

For diagonal gates on $n$ qubits the reduction of \cref{sec:capacity} gives $\capC(U)=\RoM(J_U)$ if and only if $\RoM(U\ket{+}^{\otimes n})\ge\RoM(U\ket{\mathcal{A}})$ for every proper affine subspace $\mathcal{A}$. This holds for every gate we checked at $n=3$ (seven named gates such as CCZ and CCT, and 150 random diagonal gates) and at $n=4$ (six named and eight random gates). It fails at $n=5$ for the multicontrol-T gate $C^{4}T$~\cite{SeddonCampbell2019}. An affine subspace that misses the point carrying the phase of $C^{n-1}T$ gives a stabiliser state. For an affine subspace of dimension $k<n$ through that point, $\RoM(C^{n-1}T\ket{\mathcal{A}})=\RoM(C^{k-1}T\ket{+}^{\otimes k})$ (Observation~S2 in the supplementary material of Ref.~\cite{SeddonCampbell2019}), and these values are $\sqrt2$, $1+3\sqrt2/5$, $(34+7\sqrt2)/20$ and $(178+15\sqrt2)/88$ for $k=1,\dots,4$ and $(834+31\sqrt2)/400$ for $k=5$ (Appendix~\ref{app:numerics}), or $1.414$, $1.849$, $2.195$, $2.264$ and $2.195$~\cite{SeddonCampbell2019}. The four-qubit value is the largest, so the five-qubit capacity exceeds the five-qubit Choi value and is reached by the product input $\ket{+}^{\otimes4}\ket{1}$. Seddon and Campbell conjecture that, for each family of multicontrol phase gates, the capacity is reached on $\ket{+}^{\otimes K}$ for some finite $K$ and is the same for all $n\ge K$ (their Conjecture~8.1). For $C^{n-1}T$, Table~III of their supplementary material suggests $K=4$, which would fix the capacity at $2.264$ for all $n\ge4$. The stabiliser norm $2^{-k}\sum_P|\langle P\rangle|$ of $C^{k-1}T\ket{+}^{\otimes k}$ is a lower bound on its robustness. Writing $P\propto X^{\mathbf a}Z^{\mathbf b}$, $\langle P\rangle$ vanishes for $\mathbf a=0\neq\mathbf b$, and for $\mathbf a\neq0$ its modulus is $1-(2-\sqrt2)2^{-k}$ if $\mathbf b=0$, and $(2-\sqrt2)2^{-k}$ or $\sqrt2\,2^{-k}$ if $\mathbf b\neq0$ and $\mathbf a\cdot\mathbf b$ is even or odd. Summing, with $\langle\mathds1\rangle=1$ and, for each of the $2^k-1$ vectors $\mathbf a\neq0$, $2^{k-1}$ vectors $\mathbf b$ with $\mathbf a\cdot\mathbf b$ odd and $2^{k-1}-1$ non-zero $\mathbf b$ with $\mathbf a\cdot\mathbf b$ even, gives $\mathcal{D}_k=2-(5-2\sqrt2)2^{-k}+(4-2\sqrt2)4^{-k}$, which rises from $(1+\sqrt2)/2$ at $k=1$ towards 2, so the robustness values stay above $1.99$ for $k\ge8$. For the stabiliser R\'enyi entropy $M_2$~\cite{Leone2022} the analogue of $\capC(U)=\RoM(J_U)$ fails earlier: $M_2(C^{k-1}T\ket{+}^{\otimes k})=\log_2(4/3)$, $\log_2(64/43)$ and $10-\log_2(555+168\sqrt2)$, or $0.415$, $0.574$ and $0.370$, for $k=1,2,3$, so already at three qubits a subspace input has larger $M_2$ than the Choi state.

\section{Counterexamples in odd prime dimension}
\label{app:odd}

Neither the state nor the channel class in \cref{thm:radial} can be relaxed in general. At $d=5$ the polytope of symmetric channels with $\max_vp_v\le1/5$ has 726 vertices, and 720 of them take $\rho_-$ to Wigner-positive states with $\RoM>1$ (up to $6/5$). The other six are radial, so by \cref{thm:radial} their outputs are stabiliser mixtures. At $d=7$ the symmetric channel uniform on the seven displacements $0$, $\pm(1,0)$, $\pm(0,1)$, $\pm(1,1)$ gives $\RoM=8/7$. The pure states in the $-1$ eigenspace of $A_0$ include, up to Cliffords, the states that van Dam and Howard found to tolerate the most depolarising noise~\cite{vanDamHoward2011}, but they do not inherit \cref{thm:radial}. At $d=5$ the state $(\ket1-\ket4)/\sqrt2$ under the radial channel with weight $9/50$ on each non-zero multiple of the displacement $(1,0)$ and $7/100$ on each non-zero multiple of $(1,4)$ has a non-negative Wigner function and $\RoM=1.378$, and at $d=7$ the state $(\ket1-\ket6)/\sqrt2$ with weights $3/25$ on $(1,1)$ and $7/150$ on $(1,6)$ gives $\RoM=1.482$. Along the depolarising line $(1-\gamma)\rho+\gamma\mathds{1}/d$ these states become stabiliser mixtures at $\gamma=d/(d+1)$~\cite{vanDamHoward2011}, where their Wigner function also becomes non-negative~\cite{Veitch2012}. Van Dam and Howard use Galv\~ao's notion of Wigner positivity~\cite{Galvao2005}: non-negativity in every discrete Wigner function of Gibbons \emph{et al}~\cite{Gibbons2004} at once. For this notion, non-negativity is equivalent to membership of the stabiliser polytope for mixed states as well~\cite{Cormick2006}. For the single Wigner function used here the equivalence holds for pure states but fails for some mixed states~\cite{Gross2006,Veitch2012}. We confirmed the threshold for $(\ket0-\ket1)/\sqrt2$, which is Clifford-equivalent to $(\ket1-\ket{d-1})/\sqrt2$, with exact decompositions at $d=3,5,7$, and at $d=3$ it also follows from \cref{thm:radial}.

\section{Numerical methods}
\label{app:numerics}

Robustness values and $\mu$ (\cref{eq:mu}) come from linear programs over all pure stabiliser states, generated as the Clifford orbit of $\ket{0^n}$: 6, 60, 1080 and 36\,720 states for one to four qubits. Two-qubit capacities use one representative of each of the 90 classes of three-qubit stabiliser states under Cliffords on the third qubit, and the Clifford equivalences of \cref{sec:capacity} are checked with exact arithmetic, or by a search over the 11\,520 two-qubit Cliffords modulo phases. The 17 structured gates of \cref{sec:capacity} are $T\otimes T$, CS, $\exp(i\pi ZZ/8)$, $\exp[i\pi(XX+ZZ)/8]$, $T\otimes\mathds{1}$, $\mathrm{CNOT}(T\otimes T)$, CH, controlled-$\sqrt{T}$, $\exp(i\pi XY/8)\exp(i\pi ZZ/8)$, $\sqrt{\mathrm{SWAP}}$ (equal to $\exp[-i\pi(XX+YY+ZZ)/8]$ up to a phase) and $\exp[i\pi(XX+YY)/8]$, which are Clifford-equivalent to diagonal gates, $\exp(i\pi XX/8)\exp(i\pi ZI/8)$, which is Clifford-equivalent to a single-qubit gate, and $\exp(i\pi XX/8)\exp(i\pi ZI/8)\exp(i\pi IZ/8)$, $(T\otimes T)\,\mathrm{CNOT}\,(R_H\otimes\mathds{1})$, $\mathrm{CNOT}\,(T\otimes R_H)$, $\mathrm{CS}\,(H\otimes\mathds{1})\,\mathrm{CS}$ and the controlled Hadamard-axis rotation by $\pi/8$, which are neither. Here $R_H=\exp[i\pi(X+Z)/(8\sqrt2)]$ is the Hadamard-axis rotation and CNOT has control on the first qubit. For multiplexed gates $\proj{0}\otimes V_0+\proj{1}\otimes V_1$, which include the controlled gates, the Choi state is, up to an appended $\ket{0}$, Clifford-equivalent to the output of the gate on the input with $\ket{+}$ on the control and a Bell pair between target and ancilla. This input is one of the 90, so the capacity of such a gate is the largest of 90 three-qubit robustness values. The 160 random gates of \cref{sec:capacity} are 100 controlled gates and 60 other multiplexed gates. For the permutation-invariant states $C^{k-1}T\ket{+}^{\otimes k}$, averaging an optimal decomposition over qubit permutations reduces \eqref{eq:rom} to one constraint per permutation class of Paulis (56 for five qubits), each stabiliser state entering through its Pauli expectations summed over each class. At $k=5$ the 2\,423\,520 stabiliser states give 25\,425 distinct vectors. Exact primal and dual solutions in $\mathbb{Q}(\sqrt2)$ with equal values certify the closed forms of Appendix~\ref{app:diagonal}. Exact rational primal and dual solutions certify the value $\RoM(\mathrm{CCZ}\ket{{+}{+}{+}})=23/9$ used after \cref{cor:stabA}, which Ref.~\cite{HowardCampbell2017} found numerically. The values of $\RoM(F\otimes F)$ and $\RoM(F\otimes F_\perp)$ combine the lower bound of \cref{thm:append} with exact decompositions in $\mathbb{Q}(\sqrt3)$, and the face-state equalities of Appendix~\ref{app:tightness} hold to solver precision in floating-point linear programs. The criteria of \cref{sec:clifford} are evaluated in $\Fq$ arithmetic and validated against the logical Kraus operators on the $[[4,2,2]]$, $[[5,1,3]]$ and Steane codes. The coherent-noise lifetimes of \cref{sec:memories} iterate \cref{eq:lambda}, which is checked against the logical Kraus operators of the bit-flip code. Coset sums for \cref{tab:lifetimes} are exact rational sums over all Paulis on the code ($4^9$ for the surface code), and over all Paulis of weight at most 5 for the Reed--Muller code, where the decoder is maximum likelihood for the distribution restricted to these Paulis. The omitted probability, $4.6\times10^{-9}$, moves each $\Lambda_k$ of this decoder by less than $10^{-8}$ and changes no lifetime. The checks in \cref{sec:odd} and Appendix~\ref{app:odd} use exact rational linear programs in Wigner coordinates for the channels applied to $\rho_-$ at $d=5,7$, floating-point programs for the other examples, and exact decompositions with coefficients in $\mathbb{Q}(e^{2\pi i/d})$ for the depolarising thresholds.

\end{document}